\documentclass[runningheads,envcountsame]{llncs}

\usepackage[T1]{fontenc}
\usepackage{color}
\usepackage{graphicx}
\renewcommand{\phi}{\varphi}
\renewcommand{\epsilon}{\varepsilon}

\usepackage{amsmath}
\usepackage{amssymb}
\usepackage{amsfonts}
\usepackage{mathtools}
\usepackage{thmtools, thm-restate}
\usepackage[shortcuts]{extdash}

\usepackage{tikz}
\usetikzlibrary{positioning,arrows.meta}
\usetikzlibrary{chains,fit,backgrounds}
\tikzset{
    leq/.style={
    	-{Triangle[open]},
        draw=black,
        fill=none
    },
    strictlysmaller/.style={
        -{Triangle[fill=black]},
        draw=black
    },
    strictlysmallerNew/.style={
        -{Triangle[fill=red]},
        draw=red,
        dashed
    },
    equal/.style={
        -,
        draw=black
    },
	equalNew/.style={
        -,
        draw=red,
        dashed
    },
	funcBound/.style={
        -{to},
        draw=red,
        dashed
    },
	funcBoundOld/.style={
		-{to},
		draw=black
	},
	backPath/.style={
		blue,
		dashed
	}
}
\pgfdeclarehorizontalshading{splitshading}{3cm}{%
  color(0cm)=(blue!75); 
  color(1.5cm)=(blue!75); 
  color(1.5cm)=(orange!50); 
  color(3cm)=(orange!50)
}

\usepackage{hyperref}
\usepackage[capitalise]{cleveref}

\usepackage{subcaption}
\usepackage{lipsum}
\usepackage{enumitem}
\usepackage{pdfpages}

\newcommand{\N}{\mathbb{N}}
\newcommand{\R}{\mathbb{R}}
\newcommand{\shortdash}{\scalebox{0.5}[1.0]{-}}

\newcommand{\xcopwidth}[1]{\mathrm{cw_{#1}}}

\newcommand{\vacopwidth}[1]{\mathrm{va\shortdash \xcopwidth{#1}}}
\newcommand{\cmvacopwidth}[1]{\mathrm{cm\shortdash va\shortdash \xcopwidth{#1}}}
\newcommand{\rmvacopwidth}[1]{\mathrm{rm\shortdash va\shortdash \xcopwidth{#1}}}
\newcommand{\anyvacopwidth}[1]{\mathrm{(rm\backslash cm\shortdash )va\shortdash \xcopwidth{#1}}}

\newcommand{\vlcopwidth}[1]{\mathrm{vl\shortdash \xcopwidth{#1}}}
\newcommand{\cmvlcopwidth}[1]{\mathrm{cm\shortdash vl\shortdash \xcopwidth{#1}}}
\newcommand{\rmvlcopwidth}[1]{\mathrm{rm\shortdash vl\shortdash \xcopwidth{#1}}}

\newcommand{\iacopwidth}[1]{\mathrm{ia\shortdash \xcopwidth{#1}}}
\newcommand{\cmiacopwidth}[1]{\mathrm{cm\shortdash ia\shortdash \xcopwidth{#1}}}
\newcommand{\rmiacopwidth}[1]{\mathrm{rm\shortdash ia\shortdash \xcopwidth{#1}}}
\newcommand{\anyiacopwidth}[1]{\mathrm{(rm\backslash cm\shortdash )ia\shortdash \xcopwidth{#1}}}
\newcommand{\iaNonMon}[3]{\ensuremath{G_{#2,#3}^{#1}}}

\newcommand{\ilcopwidth}[1]{\mathrm{il\shortdash \xcopwidth{#1}}}
\newcommand{\cmilcopwidth}[1]{\mathrm{cm\shortdash il\shortdash \xcopwidth{#1}}}
\newcommand{\rmilcopwidth}[1]{\mathrm{rm\shortdash il\shortdash \xcopwidth{#1}}}

\newcommand{\scol}[1]{\mathrm{scol_{#1}}}
\newcommand{\wcol}[1]{\mathrm{wcol_{#1}}}
\newcommand{\adm}[1]{\mathrm{adm_{#1}}}

\newcommand{\pw}{\mathrm{pw}}
\newcommand{\tw}{\mathrm{tw}}

\usepackage{listofitems}
\usepackage{longtable}
\newcommand{\farbigergraph}[1]{
	\setsepchar{,}
	\readlist\farblisten{#1}
	\begin{tikzpicture}[every node/.style={circle, draw, inner sep=2pt}]
			\node (a) at (-2, -2) [fill={\farblisten[1]}] {A;$K_8$};
			\node (b) at (-2, 2) [fill={\farblisten[2]}] {B;$K_8$};
			\node (c) at (2, 0) [fill={\farblisten[3]}] {C;$K_2$};
			
			\node (o1) at (-2.5, 1) [draw, circle, fill={\farblisten[4]}] {};
			\node (o2) at (-2, 1) [draw, circle, fill={\farblisten[5]}] {};
			\node (o3) at (-1.5, 1) [draw, circle, fill={\farblisten[6]}] {};
		
			\node (u1) at (-2.5, -1) [draw, circle, fill={\farblisten[7]}] {};
			\node (u2) at (-2, -1) [draw, circle, fill={\farblisten[8]}] {};
			\node (u3) at (-1.5, -1) [draw, circle, fill={\farblisten[9]}] {};
		
			\node (v) at (-2,0) [fill={\farblisten[37]}] {v};
			\node (v1) at (-1,0) [draw, circle, fill={\farblisten[10]}] {};
			\node (v2) at (0,0) [draw, circle, fill={\farblisten[11]}] {};
			\node (v3) at (1,0) [draw, circle, fill={\farblisten[12]}] {};
		
			\node (a11) at (-0.66,-2) [draw, circle, fill={\farblisten[13]}] {};
			\node (a12) at (0.66,-2) [draw, circle, fill={\farblisten[14]}] {};
			\node (a13) at (2,-1) [draw, circle, fill={\farblisten[15]}] {};
			\node (a21) at (-1,-1.5) [draw, circle, fill={\farblisten[16]}] {};
			\node (a22) at (0,-1) [draw, circle, fill={\farblisten[17]}] {};
			\node (a23) at (1,-0.5) [draw, circle, fill={\farblisten[18]}] {};
			\node (a31) at (-0.5,-1.5) [draw, circle, fill={\farblisten[19]}] {};
			\node (a32) at (0.2,-1.33) [draw, circle, fill={\farblisten[20]}] {};
			\node (a33) at (1.33,-0.66) [draw, circle, fill={\farblisten[21]}] {};
			\node (a41) at (-0.33,-1.75) [draw, circle, fill={\farblisten[22]}] {};
			\node (a42) at (0.4,-1.66) [draw, circle, fill={\farblisten[23]}] {};
			\node (a43) at (1.66,-0.75) [draw, circle, fill={\farblisten[24]}] {};
		
			\node (b11) at (-0.66,2) [draw, circle, fill={\farblisten[25]}] {};
			\node (b12) at (0.66,2) [draw, circle, fill={\farblisten[26]}] {};
			\node (b13) at (2,1) [draw, circle, fill={\farblisten[27]}] {};
			\node (b21) at (-1,1.5) [draw, circle, fill={\farblisten[28]}] {};
			\node (b22) at (0,1) [draw, circle, fill={\farblisten[29]}] {};
			\node (b23) at (1,0.5) [draw, circle, fill={\farblisten[30]}] {};
			\node (b31) at (-0.5,1.5) [draw, circle, fill={\farblisten[31]}] {};
			\node (b32) at (0.2,1.33) [draw, circle, fill={\farblisten[32]}] {};
			\node (b33) at (1.33,0.66) [draw, circle, fill={\farblisten[33]}] {};
			\node (b41) at (-0.33,1.75) [draw, circle, fill={\farblisten[34]}] {};
			\node (b42) at (0.4,1.66) [draw, circle, fill={\farblisten[35]}] {};
			\node (b43) at (1.66,0.75) [draw, circle, fill={\farblisten[36]}] {};

			\draw (a) -- (u1);
			\draw (a) -- (u2);
			\draw (a) -- (u3);
		
			\draw (b) -- (o1);
			\draw (b) -- (o2);
			\draw (b) -- (o3);
		
			\draw (v) -- (o1);
			\draw (v) -- (o2);
			\draw (v) -- (o3);
			\draw (v) -- (u1);
			\draw (v) -- (u2);
			\draw (v) -- (u3);
		
			\draw (v) -- (v1);
			\draw (v1) -- (v2);
			\draw (v2) -- (v3);
			\draw (v3) -- (c);
		
			\draw (a) -- (a11);
			\draw (a11) -- (a12);
			\draw (a12) -- (a13);
			\draw (a13) -- (c);
			\draw (a) -- (a21);
			\draw (a21) -- (a22);
			\draw (a22) -- (a23);
			\draw (a23) -- (c);
			\draw (a) -- (a31);
			\draw (a31) -- (a32);
			\draw (a32) -- (a33);
			\draw (a33) -- (c);
			\draw (a) -- (a41);
			\draw (a41) -- (a42);
			\draw (a42) -- (a43);
			\draw (a43) -- (c);
		
			\draw (b) -- (b11);
			\draw (b11) -- (b12);
			\draw (b12) -- (b13);
			\draw (b13) -- (c);
			\draw (b) -- (b21);
			\draw (b21) -- (b22);
			\draw (b22) -- (b23);
			\draw (b23) -- (c);
			\draw (b) -- (b31);
			\draw (b31) -- (b32);
			\draw (b32) -- (b33);
			\draw (b33) -- (c);
			\draw (b) -- (b41);
			\draw (b41) -- (b42);
			\draw (b42) -- (b43);
			\draw (b43) -- (c);
	
	\end{tikzpicture}
	}

\begin{document}

\title{Fast and Furious: A study on Monotonicity and Speed in Cops-and-Robber Games}

\titlerunning{Fast and Furious}

\author{Eva Fluck\inst{1}\orcidID{0000-0002-9643-6081} \and David Philipps\inst{2}\orcidID{0009-0003-0444-8507}}

\authorrunning{E. Fluck, D. Philipps}

\institute{
RWTH Aachen University, Germany\\
\email{fluck@cs.rwth-aachen.de}\\
\and
University of Oxford, UK\\
\email{david.philipps@keble.ox.ac.uk}}

\maketitle             
\begin{abstract}
	In this paper, we study different variants of the Cops-and-Robber game with respect to cop- and robber\-/monotonicity.
	We study a visible and invisible robber and variants where the robber is lazy, thus can only move when the cops announce to move on top of him.
	In all four combinations, we also vary the number $s$ of edges that the robber can traverse in a single round, called speed.
	
	We complete the study of the unbounded speed case by showing that, besides the active variants, also the visible lazy variant has both the cop- and robber\-/monotonicity property.
	Furthermore, we prove that the cop\-/monotone invisible lazy copwidth characterizes path-width, while the non\-/monotone and robber\-/monotone is known to characterize tree-width, thus these variants differ even in the unbounded speed case.
	
	We find that, even with speed restriction, the cop\-/monotone invisible copwidth and the robber\-/monotone invisible active copwidth all characterize path-width.
	On the other hand, we show that the path-width of a graph can be arbitrarily larger than the number of cops needed to win the non\-/monotone invisible active variant.
	To complete our study of cop\-/monotone variants, we show that also in the visible variants the cop\-/monotone copwidth can be arbitrarily larger than the non\-/monotone.
	
	Regarding robber\-/monotonicity, for all speeds $s\geq 4$, we give graphs where the non\-/monotone and robber\-/monotone copwidth differ.
	On the other hand, we prove that there is a function that bounds the robber\-/monotone copwidth in terms of the non\-/monotone copwidth and the speed, thus the gap between the variants is bounded.
	This proof also yields that a graph class has bounded expansion if and only if, for every speed $s$, the number of cops needed in any robber\-/monotone lazy variant is bounded by some constant $c(s)$.

\end{abstract}

\section{Introduction}
After the first study on search models in 1976 by Parsons\cite{parsons2006pursuit}, the node search model we use today under the name of Cops\-/and\-/Robber was essentially introduced in 1985 and 1986 by Kirousis and Papadimitriou \cite{kirousis1985interval,kirousis1986searching}.
The game is played in rounds on the vertices of a graph between several cops and a single robber.
The cops try to catch the robber by placing on the same vertex as him, and the robber tries to avoid getting caught.
In each round, first, some cops exit the graph and announce where they will return (they do not need to move along the edges), then the robber can move along any path that does not contain inner vertices occupied by a cop, and finally, the cops return as announced. In the bounded speed variants, we bound the maximum length of the path that the robber can use.
The copwidth measure is defined as the minimal number of cops needed to catch a smart and omniscient robber on a given graph.
It is easy to see that in a sparse graph, fewer cops may be necessary than in a clique.
Variants of the game differ, e.g. in the speed of the robber, whether he can always move or only when a cop is announced to be placed in his position, called lazy variant instead of active, or whether the robber's position is visible to the cops. We get a different game with different properties and a different copwidth measure depending on the specific ruleset.

The invisible lazy and visible active variants with unbounded speed are known to characterize tree\-/width \cite{dendris1997fugitive,seymour1993graph} and the invisible active variant path\-/width \cite{kirousis1985interval}.
An important step in these proofs often is that these game variants have certain monotonicity properties.
We differ between two monotonicity definitions.
If the cop player is restricted to never revisit a vertex, this is called cop\-/monotonicity, if the robber player is never allowed to move to a vertex that was previously occupied by a cop, this is called robber\-/monotonicity.
Intuitively this means that in a cop\-/monotone game the cops monotonously move forward through the graph, where in a robber\-/monotone game the region that the robber can reach in an unbounded speed game or the region he could still be in in an invisible game monotonously shrinks.
Our definition of robber\-/monotony slightly differs from the definition usually used in literature to translate to visible bounded speed games where there is no longer an obvious definition of the "robber\-/region" and is equivalent to the standard definition in all other cases.
We say a game variant has a monotonicity property, if the number of cops needed to catch the robber in the respective monotone variant does not differ from the non\-/monotone variant.
It is easy to see that robber\-/monotonicity is the stronger restriction on the game.
Besides the above mentioned variants of the game that have the monotonicity properties, there are also variants, that do not have these properties, for example games on directed graphs \cite{Kreutzer_digraph-decomp_2O11} and on hypergraphs \cite{Adler04}.
A monotonicity property can be helpful by restricting the search space when calculating copwidth and yield that the corresponding decision problem belongs to NP.

There are some obvious relations between the different game variants.
A strategy against an invisible robber will also work against a visible one and a strategy against an active robber also works against a lazy robber.
Furthermore a cop\-/monotone strategy is always robber\-/monotone and both are a strategy in the non\-/monotone game.
Lastly the speed restriction hinders only the robber, that is a strategy against a quick robber (even an infinite speed robber) will also work against a slow robber.
On the other hand speed is a real restriction for the robber in most variants.
Consider for example a clique with $n$ vertices where one replaces every edge by a path with $s$ inner vertices.
In the robber\-/monotone invisible lazy or cop\-/monotone visible lazy variant three cops win against a robber of any speed $\leq s$ by systematically clearing one edge after the other.
A robber of speed $>s$ however wins even the visible lazy variant against less that $n$ cops as he can always move to some free vertex of the original clique whenever the cops move.
A similar argument works for the (cop\-/monotone) visible active variant, when the edges are replaced by path with $2s-1$ inner vertices.

Search games, as the Cops\-/and\-/Robber game, have many applications.
For example they model a variety of real\-/world search problems, see \cite{FominT08} for a survey.
The number of rounds needed to capture the robber was used to show lower bounds on the runtime of the famous Weisfeiler\-/Lehman Refinement algorithm for graph isomorphism \cite{Furer_rounds_2001,Grohe2023CR-vs-WL}.
It has been proven that graph classes relating to Cops\-/and\-/Robber games, such as tree\-/width \cite{Neuen_homomorphism-distinguishing_2023}, tree\-/depth \cite{Fluck_deep-wide_2023} or bounded\-/depth tree\-/width \cite{Adler_monotone_2024}, are homomorphism distinguishing closed.
The bounded speed visible active variant is known to characterize bounded expansion \cite{torunczyk2023flip}.

\subsection{Our Contribution}
\begin{figure}[!htb]
    \crefname{theorem}{Thm.}{thms.}
    \crefname{corollary}{Cor.}{cors.}
    
    \begin{tikzpicture}
        \node (il1) at (0,0) {$\ilcopwidth{1}$};
        \node (ril1) at (2.5,0) {$\rmilcopwidth{1}$};
        \node (d*) at (5,0) {$\delta^* + 1$};
        
        \node (ils) at (0,0.8) {$\ilcopwidth{s}$};
        \node (rils) at (2.5,1) {$\rmilcopwidth{s}$};
        \node (scol) at (5,1) {$\scol{s}+1$};
        
        \node (ili) at (0,1.8) {$\ilcopwidth{\infty}$};
        \node (rili) at (2.5,1.8) {$\rmilcopwidth{\infty}$};
        \node (tw) at (5,1.8) {$\tw + 1$};

        \node (ian) at (7.5,1) {$\iacopwidth{1,s}$};
        
        \node (cili) at (0,3) {$\cmilcopwidth{1, s, \infty}$};
        \node (pw) at (2.5,3) {$\pw + 1$};
        \node (iai) at (5,3) {$\anyiacopwidth{\infty}$};
        \node (rian) at (7.5,3) {$\rmiacopwidth{s}$};
        \node (cian) at (10,3) {$\cmiacopwidth{s}$};
      
        \draw[strictlysmallerNew] (ian) -- (iai) node[midway, right] {\cref{TheoremInvActRobbermonotonicityBoundSpeed}}; 
        \draw[strictlysmaller] (d*) -- (ian);

        \draw[equalNew] (cian) -- (rian) node[midway, above] {\cref{thmt@@CorollaryInvActCopmonotonicity}};
        \draw[equalNew] (iai) -- (rian) node[midway, above] {\cref{thmt@@TheoremInvActRobbermonotonicity}};

        \draw[equal] (iai) -- (pw) node[midway, above] {\cite{kirousis1985interval,kirousis1986searching}};

        \draw[equal] (il1) -- (ril1) node[midway, above] {\cite{dendris1997fugitive}};
        \draw[equal] (d*) -- (ril1) node[midway, above] {\cite{dendris1997fugitive}};        
        \draw[strictlysmaller] (il1) -- (ils);      
        \draw[strictlysmallerNew] (ils) -- (rils) node[midway, above] {\cref{thmt@@TheoremLazyRecontamination}};
        \draw[equal] (rils) -- (scol) node[midway, above] {\cite{dendris1997fugitive}};        
        \draw[funcBound] (rils) -- (ils);        
        \draw[strictlysmallerNew] (ils) -- (ili);        
        \draw[strictlysmaller] (scol) -- (tw) node[midway, right] {\cite{nevsetvril2012sparsity}};        
        \draw[strictlysmallerNew] (ril1) -- (rils);         
        \draw[strictlysmallerNew] (rils) -- (rili);        
        \draw[equal] (ili) -- (rili) node[midway, above] {\cite{dendris1997fugitive}};        
        \draw[equal] (tw) -- (rili) node[midway, above] {\cite{dendris1997fugitive}};        
        \draw[strictlysmallerNew] (ili) -- (cili) node[midway, left] {\cref{thmt@@TheoremInvLazyCopmonotonicity}};        
        \draw[equalNew] (cili) -- (pw) node[midway, above] {\cref{thmt@@TheoremInvLazyCopmonotonicity}};
    \end{tikzpicture}

    \begin{tikzpicture}
        \node (va1) at (4.5,0) {$\vacopwidth{1}$};
        \node (rva1) at (6,0) {$\rmvacopwidth{1}$};
        \node (cva1) at (6.7,1) {$\cmvacopwidth{1}$};
        \node (vas) at (4.5,1.6) {$\vacopwidth{s}$};
        \node (adm) at (3,0.6) {$\adm{s}+1$};
        \node (scol) at (3.1,2.6) {$\scol{4s}+1$};
        \node (wcol) at (2.3,1.7) {$\wcol{2s}+1$};
        \node (cvas) at (6.7,2.4) {$\cmvacopwidth{s}$};
        \node (vai) at (4.5,3) {$\anyvacopwidth{\infty}$};
        \node (tw) at (2.2,3) {$\tw + 1$};
        
        \node (vl1) at (0,0) {$\vlcopwidth{1}$};
        \node (d*) at (1.5,0) {$\delta^* + 1$};
        \node (vls) at (0,1) {$\vlcopwidth{s}$};
        \node (ds) at (1.5,1) {$\delta^s + 1$};
        \node (rvl1) at (-2.2,0) {$\rmvlcopwidth{1}$};
        \node (cvl1) at (-4.5,0.5) {$\cmvlcopwidth{1}$};
        \node (rvls) at (-2.2,1.2) {$\rmvlcopwidth{s}$};
        \node (cvls) at (-4.5,1.4) {$\cmvlcopwidth{s}$};
        \node (vli) at (0,2.2) {$\vlcopwidth{\infty}$};
        \node (rvli) at (-2.2,2.2) {$\rmvlcopwidth{\infty}$};
        \node (cvli) at (-4.5,2.2) {$\cmvlcopwidth{\infty}$};
        \node (di) at (1.5,2.2) {$\delta^\infty + 1$};
        
        \draw [equal] (va1) -- (rva1) node[midway, above] {\cite{torunczyk2023flip}};
        \draw [equal] (va1) -- (d*) node[midway, above] {\cite{torunczyk2023flip}};
        \draw [equal] (vai) -- (tw) node[midway, above] {\cite{seymour1993graph}};

        \draw [strictlysmaller] (va1) -- (vas);
        \draw [strictlysmaller] (adm) -- (vas) node[midway, below] {\cite{torunczyk2023flip}};
        \draw [strictlysmaller] (ds) -- (vas);
        \draw [strictlysmaller] (adm) -- (ds) node[midway, below] {\cite{doi:10.1137/090780006}};
        \draw [funcBoundOld] (ds) -- (adm);
        \draw [leq] (vas) -- (scol) node[midway, above] {\cite{hickingbotham2023cop}};
        \draw [leq] (vas) -- (wcol) node[midway, above] {\cite{torunczyk2023flip}};
        \draw [strictlysmallerNew] (vas) -- (cvas) node[midway, above left=0.0cm and -0.30cm] {\cref{thmt@@TheoremVisActCopmonotonicity}};
        \draw [strictlysmallerNew] (rva1) -- (cva1) node[midway, right] {\cref{thmt@@TheoremVisActCopmonotonicity}};
        \draw [strictlysmaller] (vas) -- (vai);
        \draw [strictlysmallerNew] (cva1) -- (cvas);
        \draw [strictlysmallerNew] (cvas) -- (vai);
        \draw [strictlysmallerNew] (cvl1) -- (cvls);

        \draw [strictlysmallerNew] (cvls) -- (cvli);
        \draw [strictlysmallerNew] (rvl1) -- (rvls);
        \draw [strictlysmallerNew] (rvls) -- (rvli);

        \draw[equal] (rvl1) -- (vl1) -- (d*) node[midway, above] {\cite{doi:10.1137/090780006}};
        \draw[equal] (vls) -- (ds) node[midway, above] {\cite{doi:10.1137/090780006}};
        \draw[equalNew] (vli) -- (rvli) node[midway, above] {\cref{thmt@@TheoremVisLazyMonotonicity}};
        \draw[equalNew] (rvli) -- (cvli) node[midway, above] {\cref{thmt@@TheoremVisLazyMonotonicity}};
        \draw[equal] (vli) -- (di) node[midway, above] {\cite{doi:10.1137/090780006}};

        \draw[strictlysmaller] (di) -- (tw) node[midway, left] {\cite{doi:10.1137/090780006}};
        \draw[strictlysmaller] (vls) -- (vli) node[midway, right] {\cite{doi:10.1137/090780006}};
        \draw[strictlysmaller] (vl1) -- (vls) node[midway, right] {\cite{doi:10.1137/090780006}};
        \draw[strictlysmallerNew] (rvl1) -- (cvl1) node[midway, above] {\cref{thmt@@TheoremVisLazyCopmonotonicity}};
        \draw[strictlysmallerNew] (rvls) -- (cvls) node[midway, above] {\cref{thmt@@TheoremVisLazyCopmonotonicity}};
        \draw[strictlysmallerNew] (vls) -- (rvls) node[midway, above] {\cref{thmt@@TheoremLazyRecontamination}};
        \draw[funcBound] (rvls) -- (vls);
    \end{tikzpicture}
    \caption{Known (black) and new (red, dashed) results for the variants. Arrows with open triangle heads represent a not strict inequality, arrows with filled triangle heads represent a strict inequality on some graphs. Lines represent equalities and arrows with a tip on the reverse side represent that the corresponding inequality is functionally bounded. $s$ is an arbitrary natural number greater than 1.}
    \label{figResults}
\end{figure}
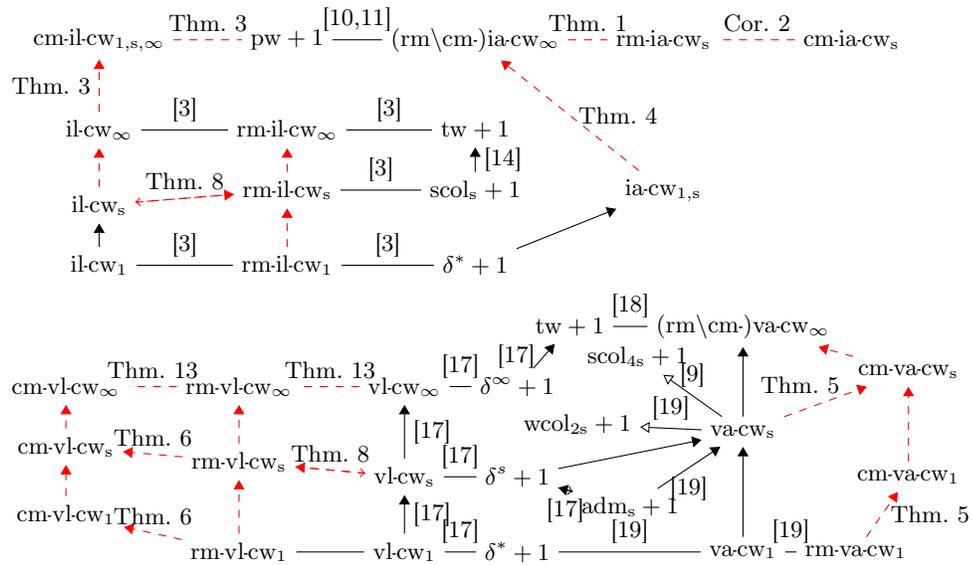
\cref{figResults} summarizes all previously known and new results.
Our work solves an open problem from Dendris, Kirousis and Thilikos \cite{dendris1997fugitive} regarding the robber\-/monotonicity property of the invisible lazy variant for $s>4$ and answers an open questions from Richerby and Thilikos \cite{doi:10.1137/090780006} regarding monotonicity of the unbounded speed visible lazy variant.

In this paper, we prove that for any speed the cop- and robber\-/monotone variants of invisible active copwidth and the cop\-/monotone variant of invisible lazy copwidth all equal path\-/width. This directly implies that the latter does not have the cop\-/monotonicity property for any speed. For the invisible active variant we provide example graphs to separate the standard bounded speed variant from path\-/width with an arbitrarily large gap and therefore show non\-/cop\-/monotonicity. The maximal size of such a gap between the monotone and the standard variant is called the monotonicity cost.

We establish the non\-/cop\-/monotonicity of the bounded speed visible variants and the unbounded monotonicity cost in these cases by constructing example graphs. Further, we give examples for the non\-/robber\-/monotonicity of the bounded speed lazy variants and provide a min\-/max characterization of the strong colouring numbers in terms of a Cops\-/and\-/Robber game. Finally we show that the unbounded speed visible lazy variant has the robber\-/monotonicity property.
\section{Preliminaries}
In this paper, we only consider finite and undirected graphs $G=(V, E)$ without self\-/loops that contain at least one edge. For $A\subseteq V$, the induced subgraph on the vertex set $V\setminus A$ is denoted by $G\setminus A$. A \emph{layout} of a graph is a linear order on the vertex set. For any $s\in \N \cup \{\infty\}$, given a layout $\prec$, a vertex $w\prec v$ is \emph{weakly $s$\=/reachable} from $v$ if there exists a path of length $\leq s$ from $v$ to $w$ such that all internal nodes $w'$ of that path satisfy $w'\succ w$. A vertex $w\prec v$ is \emph{strongly $s$\=/reachable} from $v$ if there exists a path of length $\leq s$ from $v$ to $w$ such that all internal nodes $w'$ of that path satisfy $w'\succ v$.
For a given graph, we define the following parameters as the minimal $k$ such that there is a layout where for each vertex $v$, there...
\begin{itemize}
  \item \textbf{degeneracy $\delta^*$ or width:} ...are at most $k$ vertices $w\prec v$, that are adjacent to $v$ 
  \item \textbf{weak $s$\=/coloring number $\wcol{s}$:} ...are at most $k$ vertices $w\prec v$, that are weakly $s$\=/reachable from $v$ 
  \item \textbf{strong $s$\=/coloring number $\scol{s}$ or $s$\=/elimination dimension:} ...are at most $k$ vertices $w\prec v$, that are strongly $s$\=/reachable from $v$
  \item \textbf{$s$\=/admissibility $\adm{s}$:} ... are at most $k$ paths of length $\leq s$, that start at $v$, end in some $w\prec v$ and are pairwise vertex disjoint apart from $v$ 
  \item \textbf{$s$\=/degeneracy $\delta_s$:} ...is an $A\subseteq V\setminus\{v\}$, $|A|\leq k$ such that no $w\prec v$ is reachable from $v$ in $G\setminus A$ by a path of length $\leq s$ .\footnote{In the literature these parameters are sometimes defined with a $w\preceq v$ instead of $w\prec v$ which leads to an offset by 1 in the parameter compared to the definition given above.}
\end{itemize}
Another way to characterize graph parameters is by obstructions: For $s$\=/degeneracy, for example, there is the notion of a \emph{$(k,s)$\=/hideout} as an obstruction, that is a set $U\subseteq V(G)$, such that for every vertex $v\in U$ and any set $A\subseteq V(G)\setminus \{v\}$ with $|A|<k$, there is some path from $v$ to $U\setminus \{v\}$ of length $\leq s$ in $G\setminus U$. $G$ has $\delta_s(G) \geq k$ if and only if there is a $(k,s)$\=/hideout in $G$. For more details see \cite{doi:10.1137/090780006}.
From the definitions it follows, that $\delta^*=\delta_1=\wcol{1}=\scol{1}=\adm{1}$ and for every $s\in \N$ $\adm{s} \leq \scol{s} \leq \wcol{s}$. A graph parameter $f$ is \emph{bounded in terms of} another graph parameter $g$ if there exists a function $\alpha: \R \rightarrow \R$ with $f(G)\leq \alpha (g(G))$ for every graph $G$.
Two parameters are \emph{functionally equivalent} if each is bounded in terms of the other. For example, for every $s\in\N$ $s$\=/admissibility and the generalized $s$\=/coloring numbers are pairwise functionally equivalent with a polynomial bound as shown by Dvořák in \cite{DVORAK2013833} and Nešetřil and De Mendez \cite{nevsetvril2012sparsity}.\\

\subsection{The Game}
A Cops\-/and\-/Robber game is played between a team of cops and a robber on the vertices of a given graph.
The cops try to catch the robber by moving from vertex to vertex, while the robber evades by moving along unguarded paths of the graph\footnote{Our definition is very similar to the ones used in \cite{dendris1997fugitive} and the node\-/search of \cite{doi:10.1137/090780006} with the distinction that in these papers $S_i\Delta S_{i+1}$ should have cardinality of at most 1. Intuitively, this means only one cop can be removed or placed per round. However, this is irrelevant for the visible lazy and invisible lazy variants considered in these papers, as the strategies can be easily transformed into each other: When given $S_i\Delta S_{i+1}$ with a cardinality of $k$ one could one by one firstly remove the cops, then place them on the positions where the robber is definitely not and lastly place the rest.}.
A formal Cops\-/and\-/Robber play on a Graph $G=(V, E)$ is a series of tuples $(S_0, R_0), (S_1, R_1), ...$, where $S_0=\emptyset$, $S_t$ denotes the positions of the cops and $R_t$ the position of the robber at the end of round $t$. We call max$\{|S_0|, |S_1|,...\}$ the \emph{width} of the play. The robber is caught whenever $R_t \in S_t$. He plays with \emph{speed} $s\in \N \cup \{\infty\}$ if $R_t$ and $R_{t+1}$ in $G\setminus {S_t \cap S_{t+1}}$ are connected by a path of length at most $s$. We define \emph{$\vacopwidth{s}(G)$} as the least $k$ such that there exists a \emph{strategy function} $f:\binom{V}{k} \times V \rightarrow \binom{V}{k}$\footnote{Note, that in our definition the strategy function does not depend on the history of the play.} that satisfies the following: In every play on $G$, where $S_{t+1}=f(S_t, R_t)$ and the robber is playing with speed $s$, the robber is caught and the width of the play is at most $k$ (Note that most of the time we will only give partial strategy functions that assumes a smart and omniscient robber that can be trivially be expanded to satisfy all conditions). The "va" stands for visible active, and we will name the other copwidth variants in a similar pattern. We define the \emph{invisibile active} variant of copwidth $\iacopwidth{s}$ analogously, where the strategy function can not depend on the robber position $R_t$ but just on the current cop position $S_t$ and $t$ and the \emph{lazy} variants $\vlcopwidth{s}$ and $\ilcopwidth{s}$ where only robbers need to be caught that satisfy $R_t \neq R_{t+1} \implies R_t \in S_{t+1}$. When talking about an invisible variant of copwidth, we will refer to the "possible hiding spots of the robber" at a round $t_0$ in regards to a specific cop\-/strategy $f$, meaning all the vertices $v$ for which there is a proper series of robber positions $(R_t)_{t\leq t_0}$ ending on $v$, that is not yet caught by cops acting accordingly to $f$. We say that the cops clear a vertex if they move on it, and it was previously a possible hiding spot.\\
One can see that if a graph has minimum degree $k$ and maximum degree $K$, then at least $k+1$ cops are necessary to catch a visible active or visible lazy robber and at most $K+1$ in the visible lazy case. \\ 
We are interested in the interplay between a speed bound on the robber and a monotonicity restriction on the cops. Cops are said to search a graph \emph{cop\-/monotonously} if they never revisit a vertex, so $S_i \cap S_k \subseteq S_j$ for $i<j<k \in \N$. This gives us another range of copwidth variants like $\cmvacopwidth{s}$, defined by only considering strategies leading to cop\-/monotone searches. A variant is considered to have the \emph{cop\-/monotonicity property} if the variant's copwidth is equal, regardless of whether the cops search only cop\-/monotonously. \footnote{We consistently use the wordy phrasing "the variant has the cop\-/monotonicity property" instead of saying that "the variant is cop\-/monotone" to prevent any confusion with the term "cop\-/monotone variant" when regarding the altered variant.} \\
Analogously, cops are said to search a graph \emph{robber\-/monotonously} if $R_t \notin \bigcup_{t'< t}S_{t'}$ until the robber is caught. This means that once the cops expel the robber from a particular region, he should not be able to return there later on, and gives us another range of copwidth variants like $\rmiacopwidth{s}$ etc. that are defined by only considering strategies that lead to robber\-/monotone searches. A variant is considered to have the \emph{robber\-/monotonicity property} if this restriction does not affect the corresponding copwidth. It is clear from the definition that robber\-/monotonicity is weaker than cop\-/monotonicity because any cop\-/monotone search strategy that allows a robber to enter a vertex that a cop previously occupied would not be able to catch that robber later on. Note that all of these copwidth measures are monotone with respect to the subgraph relation.\\

\section{Path-Width and Invisible Variants} 
\label{MoreNonmonotonicity}
In this section, we will consider scenarios where requiring the cops to play monotonically is a hard restriction in the sense that it makes them unable to exploit any speed bound of the robber.
The first example of such a phenomenon is the invisible active variant.
Examining the copwidth for any bounded speed and adding robber\-/monotonicity results in the same copwidth as for unbounded speed, which we know equals path\-/width.
The basic idea is that if the cops play robber\-/monotone, the robber's movement is restricted to a monotonically shrinking subset of vertices and it is not important whether he moves within this region, and therefore also not with what speed.   

\begin{restatable}{theorem}{TheoremInvActRobbermonotonicity}
    For every graph $G$ and speed $s\in\N$, it holds that
\begin{align*}
    \rmiacopwidth{s}(G) &= \rmiacopwidth{\infty}(G) \\
    &= \pw(G) + 1.
\end{align*}
\end{restatable}
Recall that cop\-/monotonicity is stronger than robber\-/monotonicity and that the invisible active variant has the cop\-/monotonicity property, so the above theorem also holds for cop\-/monotonicity.
\begin{restatable}{corollary}{CorollaryInvActCopmonotonicity}
    For every graph $G$ and speed $s\in \N$, it holds that
\begin{align*}
    \cmiacopwidth{s}(G) &= \cmiacopwidth{\infty}(G)\\
    &= \pw(G) + 1.
\end{align*}
\end{restatable}

With a similar technique we can prove the same bound for the cop\-/monotone invisible active variant.
\begin{restatable}{theorem}{TheoremInvLazyCopmonotonicity}
 For every graph $G$ and speed $s\in \N \text{ or } s=\infty$
\begin{align*}
    \cmilcopwidth{s}(G) &= \cmiacopwidth{\infty}(G)\\
    &= \pw(G) + 1.
\end{align*}
\end{restatable}
The non\-/cop\-/monotonicity of the invisible lazy copwidth now follows from the upper bound $\ilcopwidth{s}(G) \leq \ilcopwidth{\infty}(G) =\tw(G)+1$ because tree\-/width may be arbitrarily smaller than path\-/width. So $\cmilcopwidth{s}$ may be arbitrarily larger than $\ilcopwidth{s}$.

For the invisible active variant we are able to show a similar result.
\begin{restatable}{theorem}{TheoremInvActRobbermonotonicityBoundSpeed}
	\label{TheoremInvActRobbermonotonicityBoundSpeed}
	For every $g,s\in\N$ and every $n\geq 2g+2$, there is a graph $\iaNonMon sgn$ such that $\pw(\iaNonMon sgn)\geq n+g$ but $\iacopwidth{s}(\iaNonMon sgn)\leq n$.
\end{restatable}
We give the construction of \iaNonMon sgn here without proof.
\iaNonMon sgn is a disjoint union of three cliques $A,B,C$, where $|A|\coloneqq n-g$ and $|B|=|C| \coloneqq g$.
All vertices in $B$ and $C$ are connected to all vertices in $A$.
Furthermore all vertices in $B$ are connected to all vertices in $C$ via disjoint paths with $3s-2$ internal vertices.
\section{Cop-Monotonicity and the Visible Variants}
From \cite{seymour1993graph}, it is known that $\vacopwidth{\infty}$ has the cop\-/monotonicity property, but we see that this is not the case for $\vacopwidth{s}$ for finite $s \in \N$. The idea is that while searching non\-/monotonously, we can temporarily remove cops if we can guarantee that, because of his speed bound, the vertex is currently not within reach of the robber. Our constructions follow the simple blueprint of a complete $n$\=/ary tree where every node is connected by long paths to each ancestor. Playing classically the cops just need to consider the tree strucure, because they can catch the robber if he moves on a such a path. Monotonously however they need to keep the ancestors in mind, as they would need to move back on them to catch the robber on these paths. An example of such a tree can be found in \cref{fig:ExampleTree}.
\begin{figure}
	\centering
	\begin{tikzpicture}[
		every node/.style={circle, draw, inner sep=2pt},
		level distance=10mm,
		level 1/.style={sibling distance=60mm},
		level 2/.style={sibling distance=30mm},
		level 3/.style={sibling distance=15mm}]
	\node (R) {R}
	  child {node (A1) {A1}
		child {node (B1) {B1}
			child {node (C1) {C1}}
			child {node (C2) {C2}}
		}
		child {node (B2) {B2}
			child {node (C3) {C3}}
			child {node (C4) {C4}}
		}
	  }
	  child {node (A2) {A2}
		child {node (B3) {B3}
			child {node (C5) {C5}}
			child {node (C6) {C6}}
		}
		child {node (B4) {B4}
			child {node (C7) {C7}}
			child {node (C8) {C8}}
		}
	  };
	  \draw[backPath, bend left=15] (A1) to (R);
	  \draw[backPath, bend right=15] (A2) to (R);

	  \draw[backPath, bend left=35] (B1) to (R);
	  \draw[backPath, ] (B2) to (R);
	  \draw[backPath, ] (B3) to (R);
	  \draw[backPath, bend right=35] (B4) to (R);

	  \draw[backPath, bend left=50] (C1) to (R);
	  \draw[backPath, bend left=10] (C2) to (R);
	  \draw[backPath, bend left=20] (C3) to (R);
	  \draw[backPath, bend left=0] (C4) to (R);
	  \draw[backPath, bend left=0] (C5) to (R);
	  \draw[backPath, bend right=20] (C6) to (R);
	  \draw[backPath, bend right=10] (C7) to (R);
	  \draw[backPath, bend right=50] (C8) to (R);

	  \draw[backPath, bend right=15] (B1) to (A1);
	  \draw[backPath, bend left=15] (B2) to (A1);
	  \draw[backPath, bend right=15] (B3) to (A2);
	  \draw[backPath, bend left=15] (B4) to (A2);

	  \draw[backPath, bend right=15] (C1) to (A1);
	  \draw[backPath, bend left=0] (C2) to (A1);
	  \draw[backPath, bend left=0] (C3) to (A1);
	  \draw[backPath, bend left=15] (C4) to (A1);
	  \draw[backPath, bend right=15] (C5) to (A2);
	  \draw[backPath, bend left=0] (C6) to (A2);
	  \draw[backPath, bend left=0] (C7) to (A2);
	  \draw[backPath, bend left=15] (C8) to (A2);

	  \draw[backPath, bend right=15] (C1) to (B1);
	  \draw[backPath, bend left=15] (C2) to (B1);
	  \draw[backPath, bend right=15] (C3) to (B2);
	  \draw[backPath, bend left=15] (C4) to (B2);
	  \draw[backPath, bend right=15] (C5) to (B3);
	  \draw[backPath, bend left=15] (C6) to (B3);
	  \draw[backPath, bend right=15] (C7) to (B4);
	  \draw[backPath, bend left=15] (C8) to (B4);

	\end{tikzpicture}
	\caption{A similar graph to $G_4$ from \cref{thmt@@TheoremVisActCopmonotonicity}. In $G_4$ each node has 3 children, but the main ideas should become clear from this picture. Blue dashed lines represent paths of length $2s$. For \cref{thmt@@TheoremVisLazyCopmonotonicity} the length is $s+1$.}
	\label{fig:ExampleTree}
\end{figure}
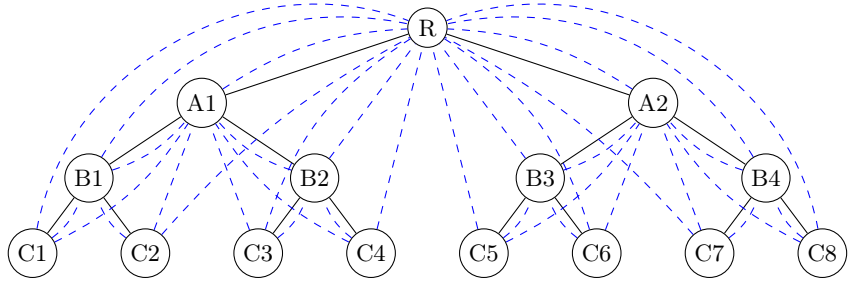

\begin{restatable}{theorem}{TheoremVisActCopmonotonicity}
For every $s\in\N$ there is a family of graphs $G_n$ with $$\vacopwidth{s}(G_n)=3 \ \ \text{and} \ \ \cmvacopwidth{s}(G_n)=n.$$
\end{restatable}
We can apply the construction with slight modifications to the length of paths to prove the same result for the visible lazy version.
\begin{restatable}{theorem}{TheoremVisLazyCopmonotonicity}
 For every $s\in \N$ there is a family of graphs $G_n$ with $$\vlcopwidth{s}(G_n)=3 \ \ \text{and} \ \ \cmvlcopwidth{s}(G_n)=n.$$
\end{restatable}
\section{Robber Monotonicity and Lazy Variants}
\label{SecitonColouringNumbers}
The bounded speed lazy variants and their robber\-/monotone analogues have a rich interaction with the strong colouring numbers. By plugging together trivial bounds and those discussed in \cite{dendris1997fugitive}\footnote{Here the authors showed the equivalence of the robber\-/monotone invisible lazy copwidth and a graph parameter called the $s$\=/elimination dimension plus one. The $s$\=/elimination dimension is defined nearly identically to the strong $s$\=/colouring number, with the only difference being that in the former, the paths that lead up in the order are counted. In contrast, the paths that lead down in the order are counted in the latter.} and \cite{torunczyk2023flip} we get the following lemma.
\begin{restatable}{lemma}{LemmaFuncEquivCaR}
    For every speed $s\in \N$, we have the relations depicted in the diagram shown in \cref{fig:LazyAndCoulorNumbers}. Moreover, it follows that those four copwidth variants, $s$\=/admissibility and the generalised colouring numbers are all pairwise functionally equivalent.
\end{restatable}

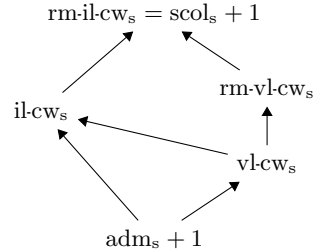
\begin{figure}
    \centering
    \begin{tikzpicture}
            \node (adm) at (0, -1) {$\adm{s}+1$};
            \node (il) at (-1.5, 0.7) {$\ilcopwidth{s}$};
            \node (vl) at (1.5, 0) {$\vlcopwidth{s}$};
            \node (rmvl) at (1.5, 1) {$\rmvlcopwidth{s}$};
            \node (scol) at (0, 2) {$\rmilcopwidth{s}=\scol{s}+1$};
    
            \draw[strictlysmaller] (adm) -- (il);
            \draw[strictlysmaller] (adm) -- (vl);
            \draw[strictlysmaller] (vl) -- (rmvl);
            \draw[strictlysmaller] (rmvl) -- (scol);
            \draw[strictlysmaller] (il) -- (scol);
            \draw[strictlysmaller] (vl) -- (il);
    \end{tikzpicture}    
    \caption{Relations between the colouring numbers and lazy copwidth variants for any $s>1$. Each arrow represents a strict inequality, note that the differences between the robber\-/monotone and non\-/monotone variants are only known for $s\geq 4$.}
    \label{fig:LazyAndCoulorNumbers}
\end{figure}

For the relationship between the invisible lazy variant and its robber\-/monotone version, it is known from \cite{dendris1997fugitive} that both are equivalent in the unbounded speed setting. This result is called "Recontamination does not help" because the potential hiding spots of an invisible robber can be thought of as a gas that spreads along unguarded edges of the graph. Then, the absence of recontamination of already cleared vertices corresponds to a robber\-/monotone cop strategy. \\
However, we get that robber\-/monotonicity is indeed not the case in the bounded speed setting, thus allowing for recontamination does help to reduce the number of cops necessary to remove the gas from the graph or, respectively, catch the invisible robber.
For all $s\geq 4$, we show this non\-/monotonicity property both for the visible lazy variant, confirming the corresponding conjecture from \cite{doi:10.1137/090780006} and the invisible lazy variant, confirming the corresponding conjecture from \cite{dendris1997fugitive}. \\
For recontamination to be a benefit to the cops, there must be a vertex $v$ that is crucial in restricting robber movement while not being a suitable hiding spot itself.
So the cops do not want to let the robber move \emph{over} $v$ but are okay with him moving \emph{on} $v$.
This allows the cops to block $v$ in some early rounds and later force the robber to move on it and catch him there.
In the unbounded speed setting, such a construction is impossible, as allowing the robber to move on a particular vertex is always connected with allowing him to move past it.
\begin{restatable}{theorem}{TheoremLazyRecontamination}
For finite $s\geq 4$, the visible lazy and invisible lazy variants of speed $s$ do not have the robber\-/monotonicity property.
\end{restatable}
The construction used in the proof is presented in \cref{fig:LazyRecontamination}. Our construction can not be trivially modified for the case $s\in \{2,3\}$ and it is questionable whether the theorem actually holds for these. The visible lazy variants of speed $1,2,3$ and $\infty$ are known to be decidable in $\mathbf{P}$, while this decision problem is $\mathbf{NP}$-hard for all other finite speeds, as shown in \cite{doi:10.1137/090780006}. This property may coincide with the robber\-/monotonicity property, but more research is needed to confirm this conjecture.  

\begin{figure}
        \centering
        \farbigergraph{white,white,white,white,white,white,white,white,white,white,white,white,white,white,white,white,white,white,white,white,white,white,white,white,white,white,white,white,white,white,white,white,white,white,white,white,white,white}
\caption{An example graph $G$ with $\vlcopwidth{4}(G)\leq\ilcopwidth{4}(G)\leq10$ and $\rmilcopwidth{4}(G)\geq\rmvlcopwidth{4}(G)>10$. $K_n$ denotes the n-clique.} 
    \label{fig:LazyRecontamination}
    \end{figure}

Thus, both lazy variants do not have the robber\-/monotonicity property for finite speeds $s\geq4$. However, in contrast to similar results of the last sections, the monotonicity cost is bounded by \cref{thmt@@LemmaFuncEquivCaR}.

At last, we want to prove that the visible lazy variant with unboudned speed has both monotonicity properties as conjectured in \cite{doi:10.1137/090780006}. To this end we define the notion of a \emph{funnel} on a graph to build our argument.
\begin{definition}
    Let $G=(V,E)$ be a graph and $S\subseteq V$. We define the \emph{cut} (you could also say neighbourhood) of $S$ in $G$ as  
    $$cut_G(S) := \{v\in V\setminus S | \ \exists s \in S \ (v,s)\in E\}.$$
    If $G$ is clear from  context we will omit the index.    
\end{definition}
\begin{definition}
    Let $G=(V,E)$ be a graph, and $A\subseteq V$. We call $F: A \rightarrow \mathcal{P}ot(A)$ a \emph{funnel} on $A$ if $\forall a \in A$:
    \begin{itemize}
        \item[(F1)] $a \in F(a)$
        \item[(F2)] $\forall b \in F(a) \setminus \{a\}: \ a \notin F(b)$
        \item[(F3)] $F(a)$ induces a conncted subgraph of $G$
    \end{itemize}

    We call $\operatorname*{max}\{|cut(F(a))| \ | \ a \in A\}$ the width of the funnel. 
    If it additionally satisfies 
    \begin{itemize}
        \item[(F4)] $\forall a \in A \forall b \in F(a) \ F(b)\subseteq F(a)$,
    \end{itemize}
    we say that the funnel is monotone.
\end{definition}
    Funnels can be seen as a similar notion to partial cop strategies in visible lazy unbounded speed games. For a given vertex $v$ in some region $A$ of the graph, a funnel $F$ assigns targets $F(v)$ in $A$ that describe where the robber is forced to move to next. $(\mathit{F1})$ demands that his current position is part of the movement. $(\mathit{F2})$ demands that the robber is not allowed to go back in the next round. $(\mathit{F3})$ demands that these movements take place in one connected subgraph. The width of the funnel plus one is then the amount of cops necessary to enforce this movement by guarding the border of the targets $F(v)$, and one cop initiating the movement by stepping on the robbers position. $(\mathit{F4})$ ensures that this is a successful and also monotone strategy on $A$ by decreasing the number of possible robber movements from one round to the next. For $A=V$ we get the following.
\begin{restatable}{lemma}{LemmaVisLazyFunnelToStrat}
        Let $G=(V,E)$ be a graph and $F: V \rightarrow \mathcal{P}ot(V)$ a monotone funnel of width $k$ then $\cmvlcopwidth{\infty}(G)\leq k+1$.
\end{restatable}

It turns out that the converse also holds, even if we drop the monotonicity requirement. We can build a funnel in steps from a given strategy by first considering vertices on which the robber can be caught in the one round. Then in each step we add a new vertex to the funnel from which the cops can force the robber to move on the existing funnel. By \cite{doi:10.1137/090780006} such vertices always exist and we show that in the framework of funnels we can guarantee the result to be monotone at each step.

\begin{restatable}{lemma}{LemmaVisLazyStratToFunnel}
    Let $G=(V,E)$ be a graph and $\vlcopwidth{\infty}(G) = k$ then there exists a monotone funnel on $V$ of width less then $k$.
\end{restatable}

We immediatly get the following.
\begin{restatable}{theorem}{TheoremVisLazyMonotonicity}
    The visible lazy variant with unbounded speed has the cop\-/monotonicity property and therefore also the robber\-/monotonicity property.
\end{restatable}
\section{Bounded Expansion}
In \cite{torunczyk2023flip}, Toruńczyk investigated the bounded speed visible active copwidth to give a new perspective on the fundamental notions of sparsity theory. A class of graphs $C$ is known to have bounded expansion iff for every $s\in \N$ the generalized $s$\=/colouring numbers are bounded. So from the bounds of these variants in the colouring numbers they concluded that the visible active copwidth variant for bounded speed can be used to characterize bounded expansion. \\
Furthermore, they adapted the results of \cite{doi:10.1137/090780006} and added the upper bound by the strong $s$\=/colouring number to infer the same for the visible lazy variant. We now take this a step further by using our results for the robber\-/monotone variants to extend Corollary A.3 from \cite{torunczyk2023flip}:
\\
\begin{restatable}{corollary}{CorollaryBoundedExpansion}
        The following conditions are equivalent for a graph class $C$:
       \begin{itemize}
        \item $C$ has bounded expansion,
        \item $\vacopwidth{s}<\infty$, for every $s\in \N$,
        \item $\vlcopwidth{s}<\infty$, for every $s\in \N$,
        \item $\rmvlcopwidth{s}<\infty$, for every $s\in \N$,
        \item $\ilcopwidth{s}<\infty$, for every $s\in \N$,
        \item $\rmilcopwidth{s}<\infty$, for every $s\in \N$.
    \end{itemize}
\end{restatable}
Planar graphs, every class of bounded maximum degree, classes of bounded tree\-/width and classes that exclude a fixed (topological) minor have bounded expansion, and thereby the above copwidth variants are all bounded on these classes.

\section{Conclusion}
While the study of unbounded speed visible active, invisible active and invisible lazy variants of the Cops\-/and\-/Robber game is characterized by monotonicity results like \cite{lapaugh1993recontamination}, the bounded speed variants give a different picture. We have studied three areas of non\-/monotonicity for bounded speed Cops\-/and\-/Robber games. In the first, we saw that cop- and robber\-/monotonicity are hard restrictions for the invisible active variant, and cop\-/monotonicity is a hard restriction for the invisible lazy variant. By that we obtained new characterizations for path\-/width. In the second, we have shown that the visible active and visible lazy variants do not have the cop\-/monotonicity property and that the monotonicity cost is unbounded.
In the third, we discovered the functional equivalence between the lazy variants and their respective robber\-/monotone versions while disproving the exact equivalence.

All of this raises the question which graph structures are captured by these formerly unknown game variants. For example one could study a variation of tree\-/width were the underlying structure is only locally acyclic. These could then further be studied for their algorithmic applications and usefulness as approximations for tree\-/width or the generalized colouring numbers. This mimics the possible use cases for visible lazy games outlined in \cite{doi:10.1137/090780006}. As the visible lazy unbounded speed variant is already known \cite{doi:10.1137/090780006} to be in $\mathbf{P}$ it is interesting which algorithmic speedup our monotonicity result (already conjectured in \cite{doi:10.1137/090780006}) offers.
The picture of the robber\-/monotonicity\-/cost still needs to be completed. For the visible lazy variant, the polynomial bound in $k$ is most probably not strict, and it would be interesting whether there may be an additive bound in terms of $s$, so whether for some function $f$, $\rmvlcopwidth{s}(G) \leq \vlcopwidth{s}(G) + f(s)$ for every graph $G$. Furthermore, the robber\-/monotonicity of both lazy copwidth variants for $s=2,3$ and the bounded speed visible active copwidth variants are open questions.

\begin{credits}
\subsubsection{\discintname}
The authors have no competing interests to declare that are
relevant to the content of this article.
\end{credits}

\bibliographystyle{splncs04}
\bibliography{references.bib}

\newpage
\appendix

\section{Appendix}

\TheoremInvActRobbermonotonicity*
\begin{proof}
    First of all, a speed bound only puts restrictions on the robber and not the cops; therefore, $\rmiacopwidth{s}(G) \leq \rmiacopwidth{\infty}(G)$. To prove the other direction, assume we have $\rmiacopwidth{s}(G)=k$ and let $f$ be a corresponding search strategy. For contradiction, assume that there is a play with a robber $(R_t)_{t\in\N}$ of infinite speed that does not get caught in it. Lets call $Z_t = \bigcup_{t'\leq t}S_{t'}$ the cleared zone at the start of round $t$. Let $t'$ be the smallest index that satisfies $R_{t'} \in Z_{t'}$ and assume it exists. Then $R_{t'} \in Z_{t'-1}$, because the robber does not get caught and $R_{t'-1} \notin Z_{t'-1}$, so there is path in $G\setminus \{S_{t'} \cap S_{t'-1}\}$ from $R_{t'-1}$ to $R_{t'}$ containing an edge $vw$ with $v\notin Z_{t'-1}$ and $w\in Z_{t'-1}$.
    So, at either round $t'-1$ or $t'$, the vertex $w$ is not occupied by a cop but has been in an earlier round; we call this round $\tilde{t}$. $v$ is not occupied for the first $t'-1$ rounds, so a robber could, regardless of its speed, wait at $v$ until round $\tilde{t}$ and then move to $w$ without getting caught, contradicting that the search strategy $f$ is robber\-/monotone for finite speeds. So no such $t'$ can exist and as $(R_t)_{t\in\N}$ does not get caught, he must remain outside $Z_t$, so $Z_t \neq V$ for all $t$. However, a speed\-/bound robber could then just stay at a vertex $v\in V\setminus \bigcup_{t}Z_{t}$ that is never visited by the cops and remains uncaught. This contradicts that $f$ is a winning cop strategy, and we get $\rmiacopwidth{\infty}(G)=k$ because our strategy is still robber\-/monotone.
\end{proof}

\TheoremInvLazyCopmonotonicity*
\begin{proof}
    $\cmilcopwidth{s}(G) \leq \cmiacopwidth{s}(G) = \pw(G) + 1$ is trivial, so we need to show $\cmiacopwidth{s}(G) \leq \cmilcopwidth{s}(G)$. Assume $\cmilcopwidth{s}(G) = k$; then there is a cop\-/monotone strategy $f$ to search for an invisible lazy robber using $k$ cops. We claim the same strategy can also search for an invisible active robber. We define $Z_t$ as above and conclude analogously that $V=\bigcup_{t}Z_{t}$ must hold. Assume for contradiction there is a play with an active robber $(R_t)_{t\in \N}$ that does not get caught, so at some round $t'$ moves along a path in $G\setminus \{S_{t'} \cap S_{t'-1}\}$ from $R_{t'-1}\in V\setminus Z_{t'-1}$ to $R_{t'}\in Z_t'$. Again we infer $R_{t'}\in Z_{t'-1}$. Then there has to be an edge $vw$ on that path, such that $v\notin Z_{t'-1}$ and $w\in Z_{t'-1}$. We have that $w\notin S_{t'-1} \cap S_{t'}$, so by the cop\-/monotonicity of our strategy $w\notin S_{t}$ for $t\geq t'$. However, a lazy robber could start at $v$ and when a cop is moved on his position, which does not happen before round $t'$, evade to $w$. There, he remains uncaught, which contradicts $f$ being a cop\-/monotone winning strategy against an invisible lazy robber, and we get $\cmiacopwidth{s}(G)=k$ because our strategy is still cop\-/monotone.
\end{proof}

\TheoremInvActRobbermonotonicityBoundSpeed*

\begin{proof}
	To see that $\pw(\iaNonMon sgn)\geq n+a$ we observe that $K_{n+g}$ is a minor of \iaNonMon san.
	
	To prove that $\iacopwidth{s}(\iaNonMon sgn)\leq n$ we give a winning\-/strategy for $n$ cops.
	For all $b\in B$ and all $c\in C$, let $P_{b,c}\coloneqq p_{b,c}^1=b,p_{b,c}^2,\ldots,p_{b,c}^{3s}=c$ be the path connecting $b$ and $c$ such that no internal vertex of the path intersects $A\cup B\cup C$.
	We start by placing cops on all vertices in $B$ and $C$.
	As $n\geq 2g+2$ there are at least 2 cops that are not yet positioned on the graph.
	These two cops clear all $P_{b,c}$, with $b\in B$ and $c\in C$.
	Now the robber can only be in a vertex of $A$.
	Next the cops stay on $C$ and move to all vertices in $A$.
	Then all $n$ cops are positioned on the graph and robber can be in any vertex $p_{b,c}^{\ell}$, where $\ell\leq s<3s$.
	In the next round the cops stay on $A$ and move to all vertices in $B$.
	Again all cops are positioned on the graph.
	The robber can be in any vertex $p_{b,c}^{\ell}$, where $1<\ell\leq 2s<3s$.
	Then the cops stay on $B$ and move to all vertices in $C$.
	The robber could move to vertices $p_{b,c}^{3s}\in C$ in this round, but would immediately be caught, thus robber can be in any vertex $p_{b,c}^{\ell}$, where $1<\ell<3s$.
	As again $|B|+|C|\leq n-2$ there are two cops that are ot positioned on the graph that can clear all these paths.
	Eventually the robber is caught.
\end{proof}

\TheoremVisActCopmonotonicity*

\begin{proof}
    First, consider the case $s=1$ and let $G_n$ be the $(n-1)$\=/ary complete tree of height $n-1$ with the addition of paths of length two from each vertex to each of its predecessors. We call the vertices of the tree "core\-/vertices" and the additional "back\-/vertices". We can search $G_n$ for a visible active robber with the following strategy: If the robber is on a back\-/vertex, we place cops on the vertex's neighbourhood and catch him, as it has a degree of two. If the robber is on a core\-/vertex, we place one cop on this vertex and one on its parent (if it exists). Then, the robber is forced to step down in the tree or leave it. If necessary, we repeat the procedure, and as he can not step down forever, he eventually leaves the tree and gets caught. To do this in a cop\-/monotone way, one would need $n$ cops who, after getting placed, do not leave their positions. \\
    On the other hand, consider $n-1$ cops that play cop\-/monotonously. The robber can start at the root vertex and step down only if a cop is placed in his position. As long as he is not at a leaf node, there is always a child node that cops do not occupy. If he reaches a leaf and a cop is placed on his position, then at least one of the $n-1$ predecessors and its corresponding back\-/vertex is not guarded by a cop (but the predecessor was prior), so he can choose this 2\=/path to go to the predecessor in two rounds and be safe. By cop\-/monotonicity, the cops cannot move back on their predecessors.
    We can adapt this for greater $s$ by replacing every edge with a fresh path of length $s$, calling the additional vertices "path\-/vertices". We always refer to the core\-/vertex and the underlying tree structure with terms like leaf-, parent-, or child\-/vertex. The robber's strategies can stay identical, as the cops have no advantage from blocking a path between two original vertices instead of blocking the vertices themselves.
    We give a non\-/cop\-/monotone strategy for three cops that distinguishes the scenarios "A" to "F" that are also depicted in \cref{fig:CopStragyTree}. In each scenario, additional vertices may already be guarded. 
       \begin{itemize}
           \item \textbf{Scenario "A"} \\
           \textit{Description:} $v$ is a core\-/vertex and the robber is either on $v$ or on a path\-/vertex between $v$ and a child of $v$.\\
           \textit{Cop Placement:} Cops are placed on $v$, its parent (if it exists) and if the robber is between $v$ and a child of $v$, on that child.\\
           \textit{Next Possible Scenarios:} The robber either gets trapped between two core\-/vertices (Scenario "F"), moves towards a back\-/vertex (Scenario "D"), or moves along the tree (Scenario "C", "E").\\
   \item \textbf{Scenario "B"} \\
           \textit{Description:} $v_1$ and $v_2$ are core\-/vertices. $v_2$ is a predecessor of $v_1$, so they are connected via a back\-/vertex and the corresponding path\-/vertices, and the robber is somewhere on that path of length $2s$.\\
           \textit{Cop Placement:} Cops are placed on $v_1$ and $v_2$. \\
           \textit{Next Possible Scenarios:} The robber either stays outside the tree (Scenario "D") or moves inside (Scenario "A"/"C"). \\
           \item \textbf{Scenario "C"} \\
           \textit{Description:} $v$ is a core\-/vertex, and the robber is either on a child of $v$ or on a path\-/vertex between $v$ and a child of $v$. A cop already guards $v$.\\
           \textit{Cop Placement:} The cop stays on $v$, and another is placed on the child.\\
           \textit{Next Possible Scenarios:} The robber either gets trapped between two core\-/vertices (Scenario "F"), moves towards a back\-/vertex (Scenario "D"), or descendents down the tree (Scenario "C", "E"). Note that Scenario "C" may only repeat finitely many times.\\
   \item \textbf{Scenario "D"} \\
           \textit{Description:} $v_1$ and $v_2$ are core\-/vertices connected via a back\-/vertex $v_b$ and the corresponding path\-/vertices. A cop already guards $v_2$, and the robber is either on the $v_b$ or a path\-/vertex between $v_2$ and the $v_b$. \\
           \textit{Cop Placement:} The cop stays on $v_2$ and additional ones are placed on $v_1$ and $v_b$. \\
           \textit{Next Possible Scenarios:} The robber gets trapped between the $v_b$ and either $v_1$ or $v_2$ (Scenario "F").\\
           \item \textbf{Scenario "E"} \\
           \textit{Description:} The robber is on a core\-/vertex $v$ that is a leaf. The parent $v_p$ is guarded by a cop.\\
           \textit{Cop Placement:} The cop stays on the parent $v_p$, and another one is placed on $v$. \\
           \textit{Next Possible Scenarios:} The robber either gets trapped between $v$ and its parent $v_p$ (Scenario "F") or moves towards a back\-/vertex (Scenario "D") \\
           \item \textbf{Scenario "F"} \\
           \textit{Description:} $v_1$ and $v_2$ are vertices that cops already guard, and the robber is on a path\-/vertex between them\\
           \textit{Arrest:} The cop on $v_2$ stays put, while the other two alternately move towards the robber and catch him. 
       \end{itemize}
   Any play with a visible active robber has to start in either Scenario "A" or "B" and, after finitely many steps, end in Scenario "F", or the robber is caught earlier. All in all, this strategy is winning for three non\-/monotone cops. It can be adapted by $n$ cop\-/monotone cops that do not leave a vertex as long as the robber is at a predecessor or descendant.
\end{proof}

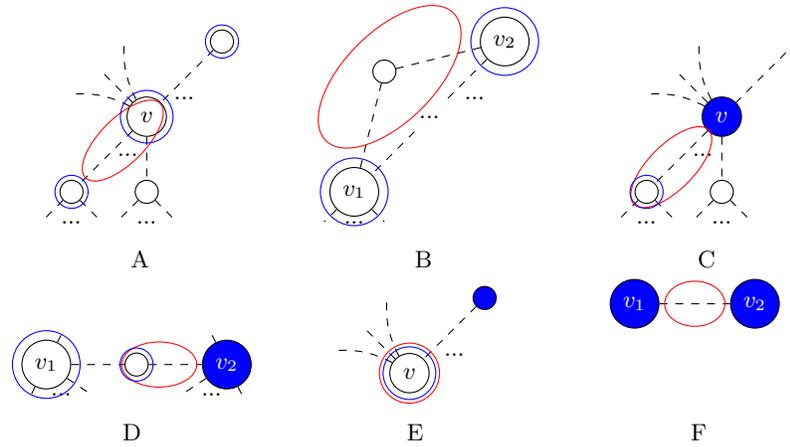
\begin{figure}
    \centering
    \begin{subfigure}{.3\textwidth}
        \centering
        \begin{tikzpicture}
            \node[draw, circle] (a) at (1,1) {};
            \node[draw, circle] (b) at (0,0) {$v$};
            \node[draw, circle] (c) at (-1,-1) {};
            \node[draw, circle] (d) at (0,-1) {};
            \node at (-0.25,-0.5) {...};
            \node at (-1,-1.4) {...};
            \node at (0,-1.4) {...};
            \node at (0.5,0.25) {...};
    
            \draw[dashed] (a) -- (b); 
            \draw[dashed] (b) -- (c); 
            \draw[dashed] (b) -- (d); 
            \draw[dashed] (b) -- (-0.6,0.6); 
            \draw[dashed, bend left=15] (b) to (-0.3,1);
            \draw[dashed, bend right=15] (b) to (-1,0.3);
            \draw[dashed] (c) to (-1.4,-1.4);
            \draw[dashed] (c) to (-0.6,-1.4);
            \draw[dashed] (d) to (-0.4,-1.4);
            \draw[dashed] (d) to (0.4,-1.4);
    
            \draw[rotate=45, red] (-0.45, 0) ellipse (0.7 and 0.3); 
            \draw[blue] (a) circle (0.22);
            \draw[blue] (b) circle (.35);
            \draw[blue] (c) circle (.22);
    
        \end{tikzpicture}
        \caption*{A}
    \end{subfigure}
    \begin{subfigure}{.3\textwidth}
        \centering
        \begin{tikzpicture}
            \node[draw, circle] (a) at (1,1) {$v_2$};
            \node[draw, circle] (b) at (-0.6,0.6) {};
            \node[draw, circle] (c) at (-1,-1) {$v_1$};
    
            \node at (0.6,0.25) {...};
            \node at (0,0) {...};
            \node at (-1,-1.4) {...};

            \draw[dashed] (a) -- (b); 
            \draw[dashed] (b) -- (c); 
            \draw[dashed] (a) -- (0.2,0.2); 
            \draw[dashed] (-0.2,-0.2) -- (c);

            \draw[dashed] (c) to (-1.4,-1.4);
            \draw[dashed] (c) to (-0.6,-1.4);
            
            \draw[rotate=45, red] (0, 0.75) ellipse (1.2 and 0.6); 
            \draw[blue] (a) circle (0.45);
            \draw[blue] (c) circle (.45);
            
        \end{tikzpicture}
        \caption*{B}
        \end{subfigure}
    \begin{subfigure}{.3\textwidth}
        \centering
        \begin{tikzpicture}
            \node[draw, circle, fill=blue] (b) at (0,0) {\textcolor{white}{$v$}};
            \node[draw, circle] (c) at (-1,-1) {};
            \node[draw, circle] (d) at (0,-1) {};
            \node at (-0.25,-0.5) {...};
            \node at (-1,-1.4) {...};
            \node at (0,-1.4) {...};
    
            \draw[dashed] (1,1) -- (b); 
            \draw[dashed] (b) -- (c); 
            \draw[dashed] (b) -- (d); 
            \draw[dashed] (b) -- (-0.6,0.6); 
            \draw[dashed, bend left=15] (b) to (-0.3,1);
            \draw[dashed, bend right=15] (b) to (-1,0.3);
            \draw[dashed] (c) to (-1.4,-1.4);
            \draw[dashed] (c) to (-0.6,-1.4);
            \draw[dashed] (d) to (-0.4,-1.4);
            \draw[dashed] (d) to (0.4,-1.4);
    
            \draw[rotate=45, red] (-0.95, 0) ellipse (0.7 and 0.3); 
            \draw[blue] (c) circle (.22);
    
        \end{tikzpicture}
        \caption*{C}
    \end{subfigure}
    \begin{subfigure}{.3\textwidth}
        \centering
        \begin{tikzpicture}
            \node[draw, circle] (a) at (-1.2,0) {$v_1$};			
            \node[draw, circle] (b) at (0,0) {};			
            \node[draw, circle, fill=blue] (c) at (1.2,0)  {\textcolor{white}{$v_2$}};
            
            \node at (-1,-0.4) {...};
            \node at (1,-0.4) {...};
    
            \draw[dashed] (a) -- (b);
            \draw[dashed] (b) -- (c);
            \draw[dashed] (a) -- (-1.4,-0.4);
            \draw[dashed] (a) -- (-0.6,-0.4);
            \draw[dashed] (a) -- (-1,0.4);
            
            \draw[dashed] (c) -- (1.4,-0.4);
            \draw[dashed] (c) -- (0.6,-0.4);
            \draw[dashed] (c) -- (1,0.4);
    
            \draw[red] (0.3, 0) ellipse (0.5 and 0.3); 
            \draw[blue] (a) circle (.45);
            \draw[blue] (b) circle (.22);
    
        \end{tikzpicture}
        \caption*{D}
    \end{subfigure}
    \begin{subfigure}{.3\textwidth}
        \centering
        \begin{tikzpicture}
            \node[draw, circle, fill=blue] (a) at (1,1) {};
            \node[draw, circle] (b) at (0,0) {$v$};
            
            \node at (0.6,0.25) {...};

            \draw[dashed] (a) -- (b); 
            \draw[dashed] (b) -- (-0.6,0.6); 
            \draw[dashed, bend left=15] (b) to (-0.3,1);
            \draw[dashed, bend right=15] (b) to (-1,0.3);
            
            \draw[blue] (b) circle (.35);
            \draw[red] (b) circle (.4);
        \end{tikzpicture}        
        \caption*{E}
    \end{subfigure}
    \begin{subfigure}{.3\textwidth}
        \centering
        \begin{tikzpicture}
            \node[draw, circle, fill=blue] (a) at (-0.8,0) {\textcolor{white}{$v_1$}};
            \node[draw, circle, fill=blue] (b) at (0.8,0) {\textcolor{white}{$v_2$}};
        
            \draw[dashed] (a) -- (b);
            \draw[red] (0, 0) ellipse (0.4 and 0.3); 
        \end{tikzpicture}
        \vspace{1cm}
        \caption*{F}
    \end{subfigure}
    \begin{subfigure}{.3\textwidth}
    \end{subfigure}
    \begin{subfigure}{.3\textwidth}
    \end{subfigure}

    \caption{The red region shows the hiding spot of the robber. The blue circles are the cop-positions after the round. The blue verteces are places where the cops were placed in the last round and stay. Each dashed line represents a path of length $s$.}
    \label{fig:CopStragyTree}
\end{figure}

\TheoremVisLazyCopmonotonicity*
\begin{proof}
    For $s=1$, we can again use the graphs $G_n$ from the last proof, as all mentioned robber strategies already are lazy. For speeds $s>1$, we remove the back\-/vertices and instead add fresh paths of length $s+1$ to connect each core\-/vertex to each of its predecessors. The robber strategy against $n-1$ cop\-/monotone cops stays the same: He starts at the root and, first, lazily descendents down the tree, taking steps of size one. Second, when he reaches a leaf and a cop is moved on his position, there is an unguarded path of length $s+1$ towards one of its predecessors. He goes into that path the maximum distance of $s$ units and is safe at the predecessor in the next round. \\
    The cop strategy is similar as well: The robber is forced down the tree by placing a cop on the parent vertex and, in the next round, placing one on the robber's position. When the robber leaves the tree, he is at a vertex of degree two and can be caught by moving in the neighbourhood of his position and, in the next round, on the vertex itself. This works with three non\-/cop\-/monotone cops or $n$ cop\-/monotone ones.
\end{proof}

\TheoremLazyRecontamination*
\begin{proof}
    We prove that the graph $G$ in \cref*{fig:LazyRecontamination} has the claimed properties of $\vlcopwidth{4}(G)\leq\ilcopwidth{4}(G)\leq10$ and $\rmilcopwidth{4}(G)\geq\rmvlcopwidth{4}(G)>10$, which implies the theorem. $G$ contains two 8\=/cliques $A$ and $B$ and a 2\=/clique $C$. The other circles represent single vertices, while the vertex $v$ plays the central role outlined in the main text.
    Edges between one of the cliques and another vertex are meant as multiple edges, connecting each vertex of the clique with that vertex. 
    For both directions it suffices to prove the stronger case.
    \subsubsection{$\rmilcopwidth{4}(G)\geq\rmvlcopwidth{4}(G)>10$}
    A visible lazy robber with speed four can force ten cops to allow recontamination. He does that by staying in $A$ until it gets completely occupied by cops. Then, if $v$ is not guarded, he can move to $B$, since the remaining two cops could not block all connecting paths of length $\leq 4$ and repeat symmetrically. Otherwise, he can escape to $C$, as there are four paths of length four connecting $C$ to $A$ and $B$ and only one cop left. \\
Next either the robber moves to $A$ or $B$ and we are were we started or the cops need to position on $C$ and all eight paths of length $\leq 4$ that connect $C$ to either $A$ or $B$. However, when they do that no cops are left and he can move on $v$ breaking the robber\-/monotonicity.

    \subsubsection{$\vlcopwidth{4}(G)\leq\ilcopwidth{4}(G)\leq10$}
    We present a strategy for ten cops on the mentioned graph, by showing the positions of the cops (blue) and remaining hiding spots of the invisible robber (red) at each round $t$.
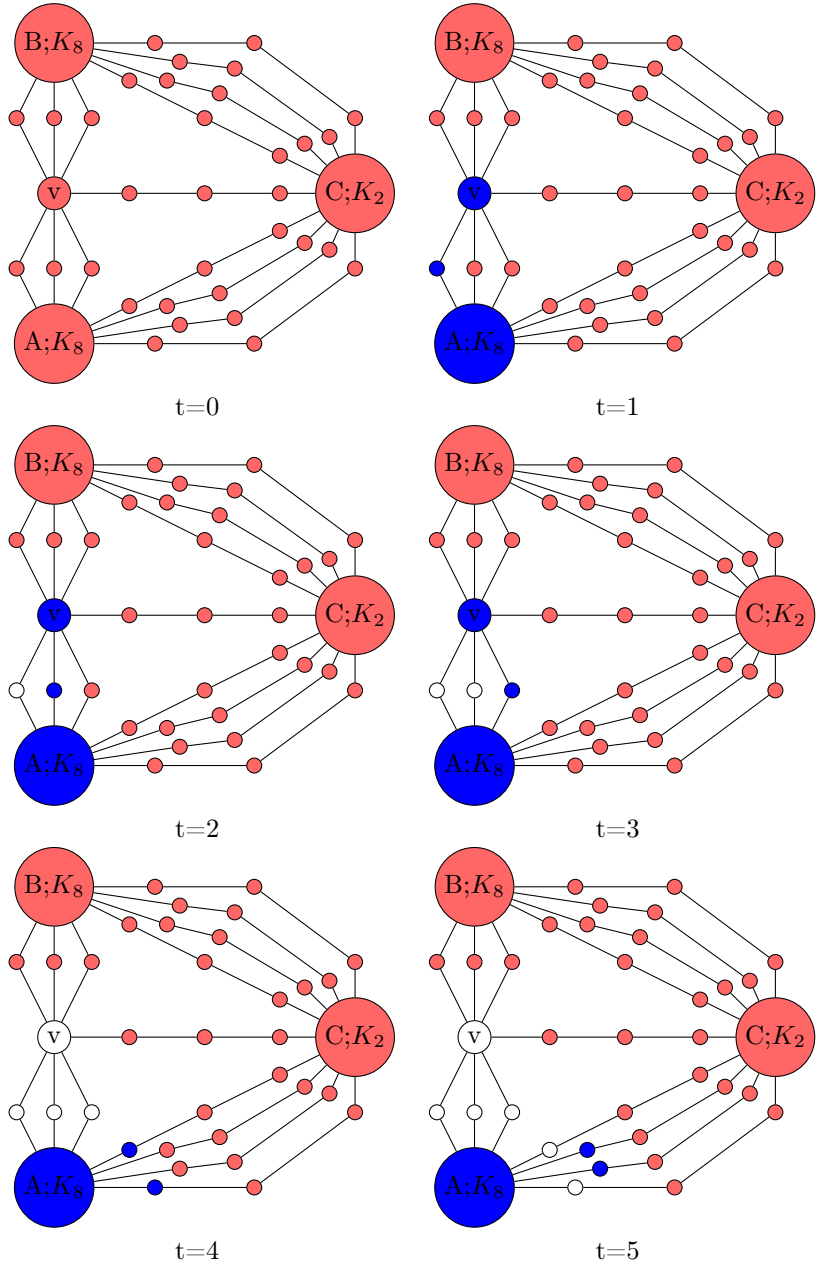
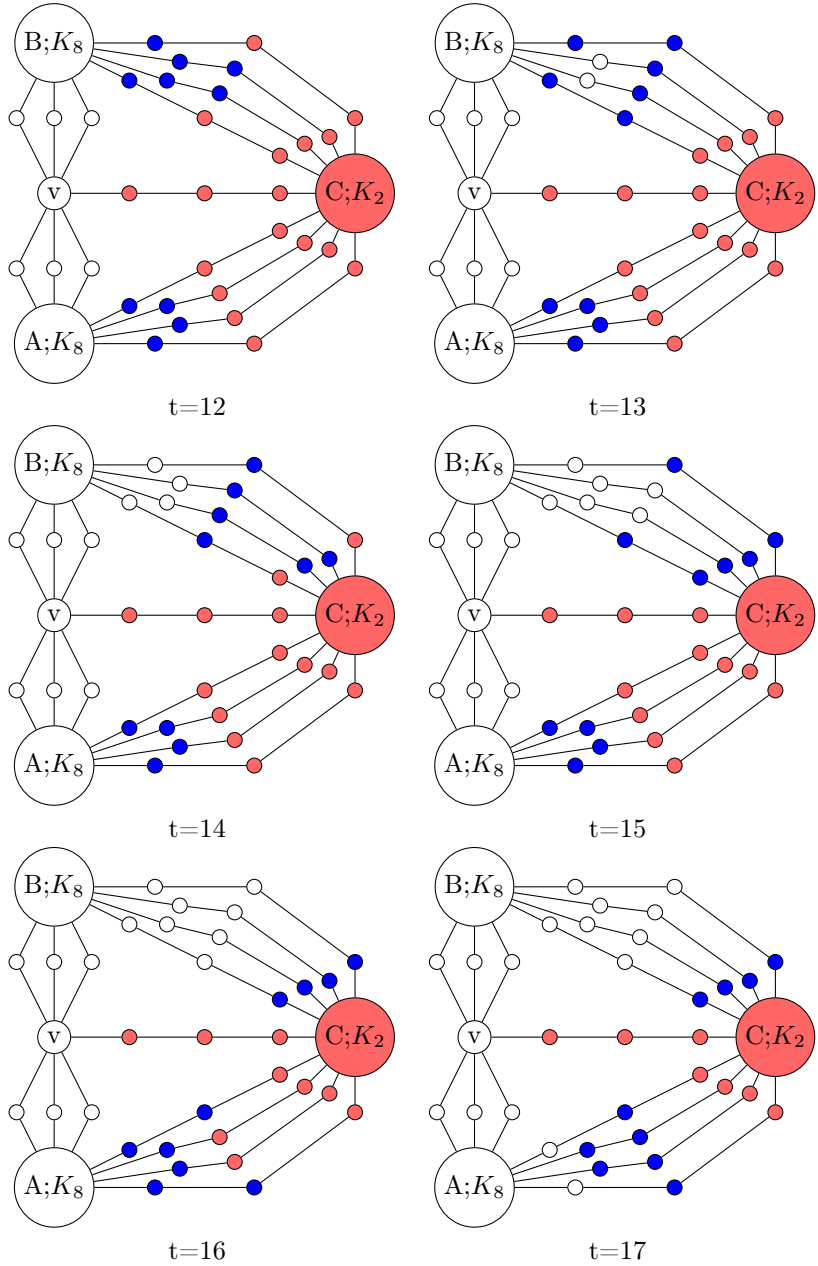
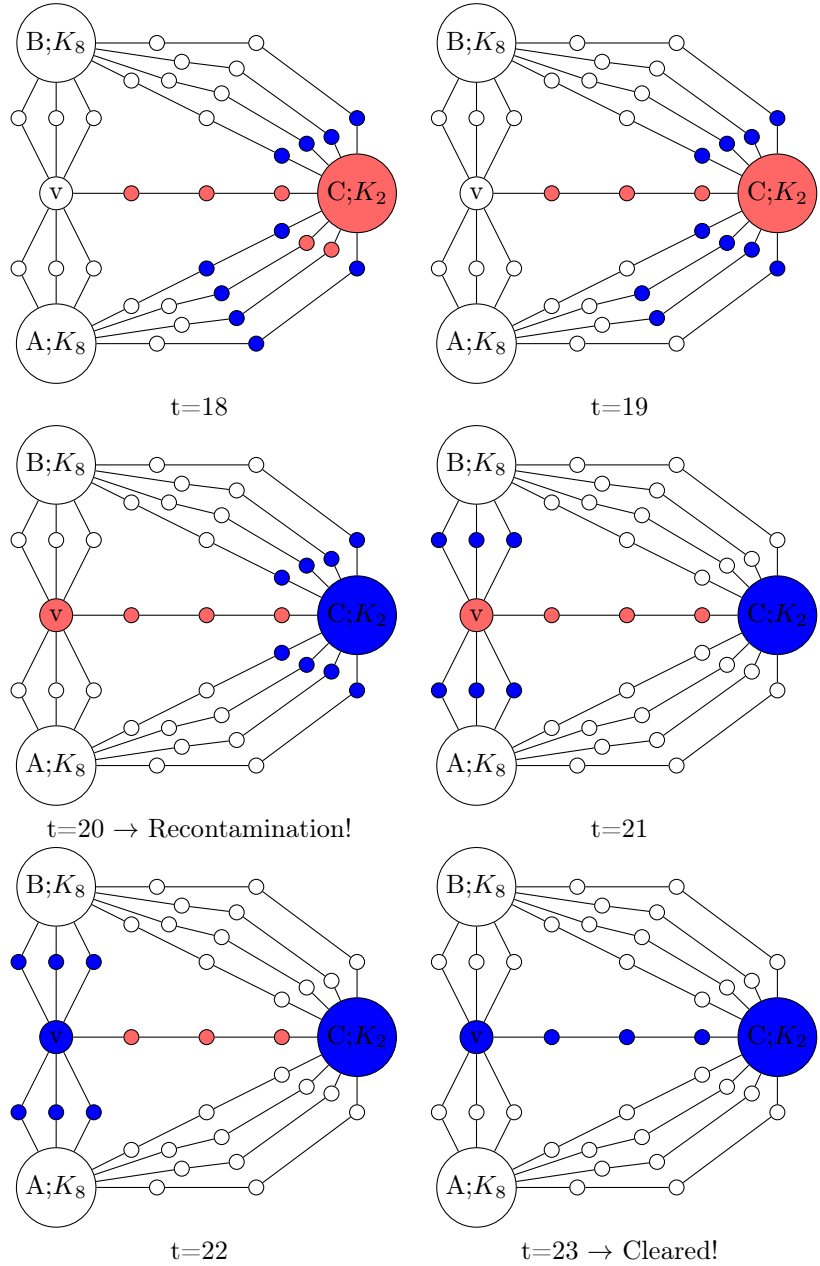
    \begin{longtable}{cc}
        \farbigergraph{red!60,red!60,red!60,red!60,red!60,red!60,red!60,red!60,red!60,red!60,red!60,red!60,red!60,red!60,red!60,red!60,red!60,red!60,red!60,red!60,red!60,red!60,red!60,red!60,red!60,red!60,red!60,red!60,red!60,red!60,red!60,red!60,red!60,red!60,red!60,red!60,red!60,red!60} &  \farbigergraph{blue,red!60,red!60,red!60,red!60,red!60,blue,red!60,red!60,red!60,red!60,red!60,red!60,red!60,red!60,red!60,red!60,red!60,red!60,red!60,red!60,red!60,red!60,red!60,red!60,red!60,red!60,red!60,red!60,red!60,red!60,red!60,red!60,red!60,red!60,red!60,blue} \\
        t=0 & t=1 \\
        \farbigergraph{blue,red!60,red!60,red!60,red!60,red!60,white,blue,red!60,red!60,red!60,red!60,red!60,red!60,red!60,red!60,red!60,red!60,red!60,red!60,red!60,red!60,red!60,red!60,red!60,red!60,red!60,red!60,red!60,red!60,red!60,red!60,red!60,red!60,red!60,red!60,blue} &  \farbigergraph{blue,red!60,red!60,red!60,red!60,red!60,white,white,blue,red!60,red!60,red!60,red!60,red!60,red!60,red!60,red!60,red!60,red!60,red!60,red!60,red!60,red!60,red!60,red!60,red!60,red!60,red!60,red!60,red!60,red!60,red!60,red!60,red!60,red!60,red!60,blue} \\
        t=2 & t=3 \\
        \farbigergraph{blue,red!60,red!60,red!60,red!60,red!60,white,white,white,red!60,red!60,red!60,blue,red!60,red!60,blue,red!60,red!60,red!60,red!60,red!60,red!60,red!60,red!60,red!60,red!60,red!60,red!60,red!60,red!60,red!60,red!60,red!60,red!60,red!60,red!60,white} &  \farbigergraph{blue,red!60,red!60,red!60,red!60,red!60,white,white,white,red!60,red!60,red!60,white,red!60,red!60,white,red!60,red!60,blue,red!60,red!60,blue,red!60,red!60,red!60,red!60,red!60,red!60,red!60,red!60,red!60,red!60,red!60,red!60,red!60,red!60,white} \\
        t=4 & t=5 \\
        \farbigergraph{white,blue,red!60,blue,red!60,red!60,white,white,white,red!60,red!60,red!60,white,red!60,red!60,white,red!60,red!60,white,red!60,red!60,white,red!60,red!60,red!60,red!60,red!60,red!60,red!60,red!60,red!60,red!60,red!60,red!60,red!60,red!60,blue} &  \farbigergraph{white,blue,red!60,white,blue,red!60,white,white,white,red!60,red!60,red!60,white,red!60,red!60,white,red!60,red!60,white,red!60,red!60,white,red!60,red!60,red!60,red!60,red!60,red!60,red!60,red!60,red!60,red!60,red!60,red!60,red!60,red!60,blue} \\
        t=6 & t=7 \\
        \farbigergraph{white,blue,red!60,white,white,blue,white,white,white,red!60,red!60,red!60,white,red!60,red!60,white,red!60,red!60,white,red!60,red!60,white,red!60,red!60,red!60,red!60,red!60,red!60,red!60,red!60,red!60,red!60,red!60,red!60,red!60,red!60,blue} &  \farbigergraph{white,blue,red!60,white,white,white,white,white,white,red!60,red!60,red!60,white,red!60,red!60,white,red!60,red!60,white,red!60,red!60,white,red!60,red!60,blue,red!60,red!60,blue,red!60,red!60,red!60,red!60,red!60,red!60,red!60,red!60,white} \\
        t=8 & t=9 \\
        \farbigergraph{white,blue,red!60,white,white,white,white,white,white,red!60,red!60,red!60,white,red!60,red!60,white,red!60,red!60,white,red!60,red!60,white,red!60,red!60,white,red!60,red!60,white,red!60,red!60,blue,red!60,red!60,blue,red!60,red!60,white} &  \farbigergraph{white,white,red!60,white,white,white,white,white,white,red!60,red!60,red!60,blue,red!60,red!60,blue,red!60,red!60,blue,red!60,red!60,blue,red!60,red!60,blue,red!60,red!60,blue,red!60,red!60,blue,red!60,red!60,blue,red!60,red!60,white} \\
        t=10 & t=11 \\
        \farbigergraph{white,white,red!60,white,white,white,white,white,white,red!60,red!60,red!60,blue,red!60,red!60,blue,red!60,red!60,blue,red!60,red!60,blue,red!60,red!60,blue,red!60,red!60,blue,red!60,red!60,blue,blue,red!60,blue,blue,red!60,white} &  \farbigergraph{white,white,red!60,white,white,white,white,white,white,red!60,red!60,red!60,blue,red!60,red!60,blue,red!60,red!60,blue,red!60,red!60,blue,red!60,red!60,blue,blue,red!60,blue,blue,red!60,white,blue,red!60,white,blue,red!60,white} \\
        t=12 & t=13 \\
        \farbigergraph{white,white,red!60,white,white,white,white,white,white,red!60,red!60,red!60,blue,red!60,red!60,blue,red!60,red!60,blue,red!60,red!60,blue,red!60,red!60,white,blue,red!60,white,blue,red!60,white,blue,blue,white,blue,blue,white} &  \farbigergraph{white,white,red!60,white,white,white,white,white,white,red!60,red!60,red!60,blue,red!60,red!60,blue,red!60,red!60,blue,red!60,red!60,blue,red!60,red!60,white,blue,blue,white,blue,blue,white,white,blue,white,white,blue,white} \\
        t=14 & t=15 \\
        \farbigergraph{white,white,red!60,white,white,white,white,white,white,red!60,red!60,red!60,blue,blue,red!60,blue,blue,red!60,blue,red!60,red!60,blue,red!60,red!60,white,white,blue,white,white,blue,white,white,blue,white,white,blue,white} &  \farbigergraph{white,white,red!60,white,white,white,white,white,white,red!60,red!60,red!60,white,blue,red!60,white,blue,red!60,blue,blue,red!60,blue,blue,red!60,white,white,blue,white,white,blue,white,white,blue,white,white,blue,white} \\
        t=16 & t=17 \\
        \farbigergraph{white,white,red!60,white,white,white,white,white,white,red!60,red!60,red!60,white,blue,blue,white,blue,blue,white,blue,red!60,white,blue,red!60,white,white,blue,white,white,blue,white,white,blue,white,white,blue,white} &  \farbigergraph{white,white,red!60,white,white,white,white,white,white,red!60,red!60,red!60,white,white,blue,white,white,blue,white,blue,blue,white,blue,blue,white,white,blue,white,white,blue,white,white,blue,white,white,blue,white} \\
        t=18 & t=19 \\
        \farbigergraph{white,white,blue,white,white,white,white,white,white,red!60,red!60,red!60,white,white,blue,white,white,blue,white,white,blue,white,white,blue,white,white,blue,white,white,blue,white,white,blue,white,white,blue,red!60} &  \farbigergraph{white,white,blue,blue,blue,blue,blue,blue,blue,red!60,red!60,red!60,white,white,white,white,white,white,white,white,white,white,white,white,white,white,white,white,white,white,white,white,white,white,white,white,red!60} \\
        t=20 $\rightarrow$ Recontamination! & t=21 \\
        \farbigergraph{white,white,blue,blue,blue,blue,blue,blue,blue,red!60,red!60,red!60,white,white,white,white,white,white,white,white,white,white,white,white,white,white,white,white,white,white,white,white,white,white,white,white,blue} &  \farbigergraph{white,white,blue,white,white,white,white,white,white,blue,blue,blue,white,white,white,white,white,white,white,white,white,white,white,white,white,white,white,white,white,white,white,white,white,white,white,white,blue} \\
        t=22 & t=23 $\rightarrow$ Cleared! \\
    
    \end{longtable}
    For higher speeds $s\geq4$, one can adjust $G$ by replacing the paths of length four between $C$ and $A$, $C$ and $B$, and $C$ and $v$ by paths of length $s$. The strategies for the cops or, respectively, the robber trivially generalise.
\end{proof}

\LemmaVisLazyFunnelToStrat*
\begin{proof}
    We give a strategy for $k+1$ cops: If the robber is on vertex $v$ then move cops on $cut(F(v))$ or if they are already positioned there move an extra cop on $v$.
    Formally, choose $$S_{i+1} = \begin{cases}
        cut(F(R_i)) \cup \{R_i\},  & \text{if } S_i = cut(F(R_i)) \\
        cut(F(R_i)),  &  \text{otherwise.}
    \end{cases}$$
    The robber will be forced to move at least every second round and by definition of the cut, every path form $F(R_i)$ to $V\setminus F(R_i)$ will be blocked by a cop in this round. It follows that $F(R_0)\supsetneq F(R_2) \supsetneq ... F(R_{2k})$ for every $k$ until $F(R_{2k}) = \{R_{2k}\}$ and the robber gets caught.
    This strategy uses $k+1$ cops. Now for the cop monotonicity we want to prove for $i<j<k$: $$S_i\cap S_k \subseteq S_j.$$
    For $v\in S_i\cap S_k$ we have that
    \begin{align*}
        v\in S_i &\implies v=R_{i-1} \vee v\in cut(F(R_{i-1})) \\
        &\implies v\notin F(R_i) \ \ \ (*) \\
        &\implies v\notin F(R_{j-1})
    \end{align*}
    we also have 
    \begin{align*}
        v\in S_k &\implies v=R_{k-1} \vee v\in cut(F_{k-1})\\
        &\overset{by (*)}{\implies} v\in cut(F(R_{k-1})) \\
        &\implies \exists u\in F(R_{k-1}): \ (v,u)\in E \\ 
        &\implies \exists u\in F(R_{j-1}): \ (v,u)\in E \\ 
    \end{align*}
    Put together we have $v\in cut(F_{j-1}) \subseteq S_j$, so the strategy induced by our monotone funnel is cop\-/monotone.
\end{proof}

\LemmaVisLazyStratToFunnel*
\begin{proof}
    We use induction over $j \leq |V|$ to prove our claim that there exists an $A_j\subseteq V$ with $|A_j|=j$ and a monotone funnel $F_j$ on $A_j$ of width less then $k$. The theorem then follows from the case $j=|V|$. Note that our proof is not constructive, even though an optimal funnel can indeed be constructed efficiently. \\
    \underline{Basecase:} For $j=0$ use $A_0=\emptyset$ and $F_j$ the empty function.\\
    \underline{Inudctive Case:} Assume the inductive hypthesis, so the existence of proper $A_j$ and $F_j$ for some $j<|V|$. Of the possible monotone funnels on $A_j$, like $F_j$, choose $F$ to be minimal in the sense, that for every other monotone funnel $F'$ on $A_j$ of width less then $k$ and $F'(a)\subseteq F(a)$ for all $a\in A_j$ we have $F=F'$. We want to work with this minimal $F$ because for the corresponding cop\-/strategy it means, that the options for the robber to move should be inclusionwise minimal. \\
    Now since $V\setminus A_j$ is not a $(k,\infty)$\=/hideout (see appendix or \cite{doi:10.1137/090780006} for definition) there is a $x\in V \setminus A_j$ and a $C\subseteq V \setminus \{x\}$ with $|C| < k$ such that $G\setminus C$ does not contain a path from $x$ to $V\setminus A$. So with this $C$ there is an extension of $F$ to the domain $A_{j+1} = A \cup \{x\}$ by setting $F(x)$ to the connected component of $x$ in $G\setminus C$ and this is a funnel of width less then $k$. For $x$ $\mathit{(F1,F3)}$ are trivial and $\mathit{(F2)}$ follows from the fact that $x\notin A_j$.
    For reasons that become clear later we choose from the possibilities of such extension of $F$ (we showed above that there is at least one) the one funnel $F_{j+1}$ that is inclusionwise maximal in regards to $F(x)$. \\
    Lets take a step back and look at what we have: A funnel $F_{j+1}$ on $A_{j+1}$ such that $F_{j+1}(x)$ is maximal, $F_{j+1}(a)$ is minimal for $a\in A_j$ and $F_{j+1}$ is already monotone on $A_j$. We claim that it is monotone on whole $A_{j+1}$:\\
    For contradiction assume that this is not the case, then there is a $y\in F_{j+1}(x)$ such that $F_{j+1}(y)\not\subseteq F_{j+1}(x)$. Let $y$ be minimal in regards to $F_{j+1}$ with that property, meaning that for $y'\in F_{j+1}(x)\cap F_{j+1}(y)$ $F_{j+1}(y')\subseteq F_{j+1}(x)$. Then let $Z$ denote the area of non\-/monotonicity $F_{j+1}(y)\cap cut(F_{j+1}(x))$ which is not empty. Corresponding to our strategy we could say that from $x$ the cops allow the robber to move to $y$ and guard the vertices $Z$ that are next to $y$, while from $y$ they then allow him the freedom to move on $Z$ but therefore have to guard some other region $\hat{Z}$. We denote this area $\hat{Z} := cut(F_{j+1}(x) \cup F_{j+1}(y)) \setminus cut(F_{j+1}(x))$.\\ 
    To get a better picture these relations are depicted in \cref{fig:VisLazyDipiction}. Our definitions of $Z$ and $\hat{Z}$ are fruitful, because they satisfy the following equations (calculations can be found further down):
    \begin{equation}
        cut(F_{j+1}(x) \cup F_{j+1}(y))=(cut(F_{j+1}(x))\setminus Z) \cup \hat{Z}\\
        \label{Calc1}
    \end{equation}
    \begin{equation}
        cut(F_{j+1}(x) \cap F_{j+1}(y))\subseteq (cut(F_{j+1}(y))\setminus \hat{Z}) \cup Z
        \label{Calc2}
    \end{equation}
    \begin{equation}
        \hat{Z} \subseteq cut(F_{j+1}(y))
        \label{Calc3}
    \end{equation}
    
    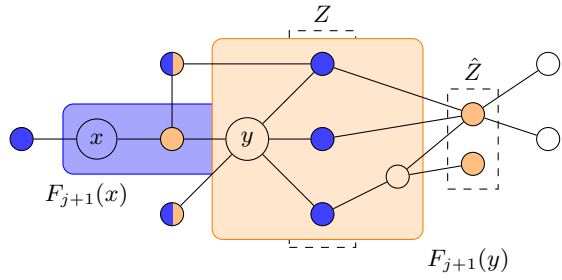
\begin{figure}
    \centering
    \begin{tikzpicture}
        \pgfdeclarelayer{background}
        \pgfdeclarelayer{main}
        \pgfsetlayers{background,main}

        \node[draw, circle] (x) at (0,0) {$x$};
        \node[draw, circle, fill=blue!75] (A) at (-1,0) {};
        \node[draw, circle, shading=splitshading] (B) at (1,1) {};
        \node[draw, circle, fill=orange!50] (M) at (1,0) {};
        \node[draw, circle] (Y) at (2,0) {$y$};
        \node[draw, circle, shading=splitshading] (C) at (1,-1) {};
        \node[draw, circle, fill=blue!75] (Z1) at (3,1) {};
        \node[draw, circle, fill=blue!75] (Z2) at (3,0) {};
        \node[draw, circle, fill=blue!75] (Z3) at (3,-1) {};
        \node[draw, circle] (F) at (4,-0.5) {};
        \node[draw, circle, fill=orange!50] (ZH1) at (5,0.33) {};
        \node[draw, circle, fill=orange!50] (ZH2) at (5,-0.33) {};
        \node[draw, circle] (D) at (6,1) {};
        \node[draw, circle] (E) at (6,0) {};

        \draw (A) -- (x)
                (x) -- (M)
                (B) -- (M)
                (B) -- (Z1)
                (M) -- (Y)
                (Y) -- (C)
                (Z1) -- (Y)
                (Z2) -- (Y)
                (Z3) -- (Y)
                (Z3) -- (F)
                (Z1) -- (ZH1)
                (Z2) -- (ZH1)
                (F) -- (ZH1)
                (F) -- (ZH2)
                (ZH1) -- (D)
                (ZH1) -- (E);
    \begin{pgfonlayer}{background}
        \node[draw, dashed, fit=(ZH1) (ZH2), inner sep=5pt, label=above:{\(\hat{Z}\)}] {};

        \node[draw, dashed, fit=(Z1) (Z2) (Z3), inner sep=8pt, label=above:{\(Z\)}] {};

        \node[draw=blue,rounded corners, fill=blue!35, fill opacity=0.5, fit=(x) (M) (Y), inner sep=5pt, label=below left:{\(F_{j+1}(x)\)}] {};

        \node[draw=orange,rounded corners, fill=orange!20, fill opacity=0.5, fit=(Y) (Z1) (Z2) (Z3) (F), inner sep=5pt,label=below right:{\(F_{j+1}(y)\)}] {};
    \end{pgfonlayer}
    \end{tikzpicture}

    \caption{Depiction of the different elements in the proof of \cref{thmt@@LemmaVisLazyStratToFunnel}. The dashed boxes represent the sets \(\hat{Z}\) and \(Z\), the light blue box represents the set \(F_{j+1}(x)\), and the orange box represents the set \(F_{j+1}(y)\). The blue and orange colored vertices are the vertices in the sets \(cut(F_{j+1}(x))\) and \(cut(F_{j+1}(y))\), respectively.}
    \label{fig:VisLazyDipiction}
\end{figure}

    With all of this build up we can make a case distinction on the sizes of $Z$ and $\hat{Z}$ and derive contradictions in both cases.\\
    \subsubsection{Case $|\hat{Z}| \leq |Z|$:}
    We have that $$\hat{F}: A_{j+1} \rightarrow \mathcal{P}ot(A_{j+1}), v \mapsto \begin{cases}
        F_{j+1}(v) & \text{if } v \in A_j\\
        F_{j+1}(x) \cup F_{j+1}(y) & \text{if } v=x
    \end{cases}$$
    is also a funnel. Its width is less then $k$ since 
    
    \begin{align*}
        |cut(\hat{F}(x))| &\overset{\cref*{Calc1}}{=} |\left(cut(F_{j+1}(x))\setminus Z \right) \cup \hat{Z}|\\
        &= |cut(F_{j+1}(x))| - |Z| + |\hat{Z}|\\
        &\leq |cut(F_{j+1}(x))|\\
        &< k,
    \end{align*}
    
    but $\hat{F}(x)\supsetneq F_{j+1}(x)$ which contradicts that $F_{j+1}(x)$ was maximal.\\

    \subsubsection{Case $|\hat{Z}| > |Z|$:}
    We have that $$\hat{F}: A \rightarrow \mathcal{P}ot(A), a \mapsto \begin{cases}
        F_{j+1}(a) & \text{if } a \neq y\\
        F_{j+1}(x) \cap F_{j}(y) & \text{if } a=y
    \end{cases}$$
    is also a target funnel. Its width is less then $k$ since 

    \begin{align*}
        |cut(\hat{F}(y))| &\overset{\cref*{Calc2}}{\leq} |\left(cut(F_{j+1}(y))\setminus \hat{Z}\right) \cup Z|\\
        &\leq |cut(F_{j+1}(y)) \setminus \hat{Z}| + |Z|\\
        &\overset{\cref*{Calc3}}{=} |cut(F_{j+1}(y))| - |\hat{Z}| + |Z|\\
        &< |cut(F_{j+1}(y))|\\
        &< k
    \end{align*}

    and it is moreover monotone since 
    $$\forall a \in A \setminus \{y\} \forall b \in \hat{F}(a): \ \hat{F}(b) \subseteq F_{j+1}(b) \subseteq F_{j+1}(a) = \hat{F}(a)$$
    
    and for all $b\in \hat{F}(y)=F_{j+1}(x)\cap F_{j+1}(y)$ the minimality of $y$ implies $\hat{F}(b) \subseteq F_{j+1}(x)$ and the monotonicity of $F$ implies $\hat{F}(b) \subseteq F_{j+1}(y)$. Thus $\hat{F}(b) \subseteq F_{j+1}(x)\cap F_{j+1}(y)=\hat{F}(y)$. But $\hat{F}(y)\subsetneq F_{j+1}(y)$ which contradicts that $F_{j+1}$ and therefore $F$ was minimal on $A_j$. \\

    So overall $F_{j+1}$ has to be monotone on $A_{j+1}$ and we have proven the inductive step. The only thing left are the three equations from earlier:\\
\underline{Claim $cut(F_{j+1}(x) \cup F_{j+1}(y))=(cut(F_{j+1}(x))\setminus Z) \cup \hat{Z}$:}\\
\begin{align*}
    \left(cut(F_{j+1}(x))\setminus Z\right) \cup \hat{Z} =&\ \left(cut(F_{j+1}(x))\setminus F_{j+1}(y)\right) \cup \hat{Z} \\
    =&\ \left\{v \in  V\setminus \left( F_{j+1}(x)\cup F_{j+1}(y)\right) \ | \ \exists t \in F_{j+1}(x): \ (t,v)\in E  \right\} \cup \hat{Z} \\
    =&\ \left(cut\left(F_{j+1}(x) \cup F_{j+1}(y)\right) \cap cut\left(F_{j+1}(x)\right)\right) \\
    & \cup (cut(F_{j+1}(x) \cup F_{j+1}(y)) \setminus cut(F_{j+1}(x)))\\
    =&\ cut(F_{j+1}(x) \cup F_{j+1}(y))
\end{align*}
\underline{Claim $cut(F_{j+1}(x) \cap F_{j+1}(y))\subseteq (cut(F_{j+1}(y))\setminus \hat{Z}) \cup Z$:}\\
\begin{align*}
    (cut(F_{j+1}(y))\setminus \hat{Z}) \cup Z =& \left(cut(F_{j+1}(y))\setminus cut(F_{j+1}(x) \cup F_{j+1}(y))\right) \\
    &\cup \left(cut(F_{j+1}(y)) \cap cut(F_{j+1}(x))\right) \\
    &\cup Z \\
    =& \{v \in F_{j+1}(x)\setminus F_{j+1}(y) \ | \ \exists t \in F_{j+1}(y): \ (t,v)\in E\} \\
    &\cup \{v \in V\setminus (F_{j+1}(y) \cup F_{j+1}(x)) \ | \ \exists t \in F_{j+1}(y): \ (t,v)\in E \\
    &\hspace{5cm}                                        \wedge \exists t \in F_{j+1}(x): \ (t,v)\in E\}\\
    &\cup \{v \in F_{j+1}(y)\setminus F_{j+1}(x) \ | \ \exists t \in F_{j+1}(x): \ (t,v)\in E\} \\
    \supseteq& \{v \in F_{j+1}(x)\setminus F_{j+1}(y) \ | \ \exists t \in F_{j+1}(x)\cap F_{j+1}(y): \ (t,v)\in E\}\\
    &\cup \{v \in V\setminus (F_{j+1}(y) \cup F_{j+1}(x)) \ | \ \exists t \in F_{j+1}(x) \cap F_{j+1}(y): \ (t,v)\in E \}\\
    &\cup \{v \in F_{j+1}(y)\setminus F_{j+1}(x) \ | \ \exists t \in F_{j+1}(x) \cap F_{j+1}(y): \ (t,v)\in E\}\\
    =& cut(F_{j+1}(x) \cap F_{j+1}(y))
\end{align*}
\underline{Claim $\hat{Z} \subseteq cut(F_{j+1}(y))$}
\begin{align*}
    \hat{Z} =& cut(F_{j+1}(x) \cup F_{j+1}(y)) \setminus cut(F_{j+1}(x))\\
    \subseteq& \{v \in V\setminus \left(F_{j+1}(x) \cup F_{j+1}(y)\right) \ | \ \exists t \in \left(F_{j+1}(x) \cup F_{j+1}(y)\right): \ (t,v)\in E\} \\
    &\setminus \{v \in V\setminus \left(F_{j+1}(x) \cup F_{j+1}(y)\right) \ | \ \exists t \in F_{j+1}(x): \ (t,v)\in E\}\\
    \subseteq& \{v \in V\setminus F_{j+1}(y) \ | \ \exists t \in F_{j+1}(y): \ (t,v)\in E\} \\
    =& cut(F_{j+1}(y))
\end{align*}
\qed
\end{proof}
\end{document}